\documentclass[aps,pra,twocolumn,showpacs,superscriptaddress,nofootinbib]{revtex4-2}   
\usepackage{graphicx}  
\usepackage[caption=false,subrefformat=parens,labelformat=parens]{subfig}
\usepackage{float}
\usepackage[toc,page]{appendix}
\usepackage[table]{xcolor}
\usepackage{mathtools}
\usepackage{dcolumn}   
\usepackage{bm}        
\usepackage{color}

\usepackage{hyperref}
\usepackage{pifont}
\usepackage{mathrsfs}
\usepackage{amsmath, amsthm, amssymb}
\usepackage{amsfonts}
\usepackage{relsize}    
\usepackage{bbold}       
\usepackage{accents}    

\hypersetup{
	colorlinks=true, 
	linktoc=all,     
	linkcolor=blue,  
	citecolor=blue,
	filecolor=blue,
	urlcolor=blue
}
\usepackage{soul}

\usepackage{tikz}
\usetikzlibrary{arrows.meta, decorations.markings}
	
\newcommand{\beq}{\begin{equation}}
	\newcommand{\eeq}{\end{equation}}
\newcommand{\bqa}{\begin{eqnarray}}
	\newcommand{\eqa}{\end{eqnarray}}

\newcommand{\erf}[1]{Eq.~(\ref{#1})}

\newcommand{\erfs}[2]{Eqs.~(\ref{#1})--(\ref{#2})}
\newcommand{\erfa}[2]{Eqs.~(\ref{#1}) and (\ref{#2})}
\newcommand{\arf}[1]{{App.}~\ref{#1}} 
 
\newcommand{\crf}[1]{Ref.~\cite{#1}} 
\newcommand{\trf}[1]{Table ~\ref{#1}}

\newcommand{\frf}[1]{Fig.~\ref{#1}}

\newcommand{\ie}{{\it i.e.}}
\newcommand{\dg}{^\dagger}

\definecolor{BLACK}{gray}{0}
\definecolor{RED}{rgb}{1,0,0}
\definecolor{GREEN}{rgb}{0.2,.6,0.2}
\definecolor{amber}{rgb}{1.0,0.49,0.0}

\renewcommand{\[}{\left[}

\newcommand{\smallfrac}[2]{\mbox{$\frac{#1}{#2}$}}

\newtheorem*{definition}{Definition}
\newtheorem*{proposition}{Proposition}

\newcommand\Tstrut{\rule{0pt}{3.5ex}}

\begin{document}
	
	\widetext

	
	\title{Nonnormality and Dissipation in Markovian Quantum Dynamics: Implications for Quantum Simulation}

	\author{Shakib Daryanoosh} 
\email{shakib.daryanoosh@curtin.edu.au}
\affiliation{Curtin Centre for Optimisation and Decision Science, Curtin University, Whadjuk Country, Perth 6102, Australia}

%
%


%
\vskip 0.25cm
	\date{\today}

\begin{abstract}
	Understanding the structure and stability of open quantum dynamics is increasingly important for both fundamental studies of nonequilibrium quantum systems and the development of quantum simulation algorithms. In this work, we introduce a structural framework for Markovian open quantum systems that characterizes Lindbladian generators in terms of two scalar quantities: the dissipative strength and the nonnormality. 
	We show that normal generators admit an exact decoupling between dissipative and norm-preserving dynamics, leading to purely exponential behavior governed by the dissipative scale. In contrast, nonnormality is an intrinsically dissipative feature: it vanishes in the absence of dissipation but is not implied by it. Moreover, it is structurally constrained by the interplay between the Hermitian and anti-Hermitian components of the generator. For generic Markovian open quantum systems, we identify  scaling  regimes controlled by a dimensionless ratio between nonnormality and dissipative strength, governing the onset of transient amplification.
	These structural features have direct implications for quantum simulation. While Hamiltonian and normal dissipative dynamics exhibit stable evolution with standard scaling behavior, nonnormal generators can induce transient growth that amplifies numerical errors and increases simulation cost. Our results provide a unified generator-level perspective on irreversibility, stability, and quantum simulation of open quantum systems.
\end{abstract}

	\maketitle

\section{Introduction} \label{sec:intro}
Open quantum systems are foundational to understanding irreversibility, decoherence, and the emergence of classical behavior from quantum mechanics~\cite{BrePet02, RivHue12, GarZol04}. Beyond fundamental interest, open-system dynamics are also of pressing technological importance: simulating open quantum systems is a promising application for quantum computers~\cite{Weimer15, SchMaz21, KamMin22,SurMar23,PocWie25,SanMen25,LiuGuo25}. Moreover, concepts of dissipation and decoherence are intimately linked to strategies in quantum error correction and fault tolerance~\cite{LidBru13,Terhal15, LooKho18,LieGor20}. Understanding and controlling these processes is therefore essential for both fundamental insights and the practical realization of quantum technologies.

Most theoretical studies characterizing Markovian open quantum systems have focused on characterizing dynamics through spectral properties and algebraic structures of the generator, with particular emphasis on steady states and asymptotic convergence~\cite{Spohn78,WolCir08,ZanRas97,LidWha98,SasUma23,LiPol23,Yoshida24,LiZhe25}. While these approaches provide powerful tools for analyzing long-time behavior, relatively limited attention has been given to systematic classifications of dissipative dynamics based on intrinsic structural features of the generator~\cite{BauThi08,BauNar12,AlbJia16}. 

 Open-system dynamics has been studied from several complementary perspectives, including  quantum simulation, nonnormal operator theory, and non-Hermitian quantum physics, each emphasizing different structural aspects of the generator. Algorithmic studies of Lindbladian simulation focus on complexity bounds based on spectral properties or operator norms~\cite{ChiZhu21,KliEis11,BerSom15}, while the theory of nonnormal operators highlights the role of pseudospectra and transient growth in governing short-time dynamics and sensitivity~\cite{TreEmb05,ChaMeh98, NavLar26}. Related effects also arise in non-Hermitian quantum systems, where eigenmode nonorthogonality underlies enhanced response and dissipative phase transitions~\cite{AshUed20,BelDon25}.

Taken together, these approaches emphasize complementary aspects of open-system dynamics but remain largely disconnected. In particular, the interplay between dissipative strength and nonnormality in physically relevant generators has not been systematically captured within a unified and operational framework, which motivates the approach developed in this work.

In this paper, we introduce a structural framework for Markovian quantum dynamical generators that formalizes this relationship in a minimal and quantitative way using two scalar quantities: the dissipative strength, $\delta(\mathcal L)$, which quantifies the overall magnitude of irreversible effects, and the nonnormality, $\eta(\mathcal L)$, which captures directional flow in operator space and short-time transient phenomena~\cite{TreEmb05}. We show that these quantities are not independent, but obey structural constraints that can restrict the accessible dynamical regimes. By decomposing the generator into Hermitian and anti-Hermitian components with respect to the Hilbert–Schmidt inner product, our approach naturally distinguishes between dissipative and oscillatory contributions to the dynamics. 

The resulting $(\delta, \eta)$--classification admits a simple geometric representation that highlights three  benchmark  regimes: purely Hamiltonian dynamics, normal dissipative evolutions, and generic nonnormal open systems exhibiting transient amplification. This framework not only provides insight into qualitative differences between these regimes but also informs the scaling of simulation costs and sensitivity to perturbations.

We show that for normal Lindbladians, the dynamics admits an exact decoupling into dissipative and rotational (norm-preserving) components, leading to purely exponential behavior governed by the dissipative scale. In contrast, nonnormality emerges as a genuinely dissipative phenomenon: on one hand it cannot arise in the absence of dissipation, and on the other hand it is not implied by it. Instead it reflects a finer  structural property in operator space  of the generator associated with noncommutativity between its components. 

Building on this distinction, we identify  scaling  regimes in which nonnormal effects are either perturbative, comparable to, or dominant over dissipative decay, thereby organizing open-system dynamics into weakly nonnormal, crossover, and strongly nonnormal regimes. This classification has direct implications for quantum simulation. While normal and weakly nonnormal dynamics retain the favorable scaling of closed-system evolution up to constant factors, stronger nonnormality leads to transient amplification and enhanced sensitivity to errors, which can increase simulation cost by effectively coupling precision requirements to dynamical properties of the generator.

While a full classification of generic Lindbladians remains open, 
our approach highlights the key structural features that determine simulation complexity. Here and throughout, we use the term computational complexity
in an operational sense, referring to the scaling of resources required
to simulate the dynamics (such as time, precision, or numerical stability),
rather than in a strict oracle-based query complexity model.
Our approach complements existing studies of open quantum systems by providing a generator-level perspective that is independent of specific states or steady-state assumptions. We demonstrate the utility of the framework through representative examples.

The remainder of the paper is organized as follows. In Sec.~\ref{sec:formalism}, we introduce the formalism and define the key metrics along with structural properties of quantum generators. Here we present an operational interpretation of the measures and their link to quantum simulation cost. Sec.~\ref{sec:classes} illustrates the classification with  representative  examples from Hamiltonian, normal dissipative, and nonnormal regimes. Section~\ref{sec:conclusion} provides a summary of results and outlook for future work. Technical details and proofs are provided in the appendices.

\section{Structural Framework for Quantum Dynamical Generators} \label{sec:formalism}
We consider Markovian open quantum systems described by the Gorini--Kossakowski--Sudarshan--Lindblad master equation~\cite{Lindblad76,GKS76}
\begin{equation} \label{eqn:GKSL}
	\dot{\rho} = -i[\hat H,\rho] + \sum_k \left( L_k \rho L_k^\dagger - \frac{1}{2} \{L_k^\dagger L_k, \rho\} \right) \equiv \mathcal L(\rho),
\end{equation}
where $\mathcal L$ is the Lindbladian generator acting on operator space, and $\{\hat L_k\}$ are Lindblad (system) operators\footnote{In this work we restrict to finite-dimensional Hilbert spaces, so that all operators and superoperators are bounded and operator norms are well-defined.}. Throughout, we analyze structural properties directly at the level of $\mathcal L$, independent of particular states or steady-state assumptions, and equipped with the Hilbert–Schmidt operator structure that provides a natural geometry for the space of operators.
Any generator $\mathcal L$ admits the decomposition\footnote{This decomposition is  defined  with respect to the Hilbert–Schmidt inner product and is uniquely suited to analyzing norm growth and error amplification in quantum simulation.}
\begin{equation} \label{L:d:nd}
	\mathcal L = \mathcal L_{\rm d} + \mathcal L_{\rm nd},
\end{equation}
where
\begin{equation} \label{Ld:Lnd}
	\mathcal L_{\rm d} := \frac{1}{2}(\mathcal L + \mathcal L^\dagger), \qquad 
	\mathcal L_{\rm nd} := \frac{1}{2}(\mathcal L - \mathcal L^\dagger).
\end{equation}
Here $\mathcal L_{\rm d}$ is Hermitian and $\mathcal L_{\rm nd}$ is anti-Hermitian with respect to the Hilbert--Schmidt inner product\footnote{This choice ensures basis independence and aligns with standard formulations of quantum dynamical maps. Alternative inner products would lead to modified decompositions. The adjoint of a superoperator ${\cal S}\dg$ is defined according to $\langle \hat A,{\cal S}(\hat B)\rangle = \langle{\cal S}\dg(\hat A),\hat B\rangle$.} (for operators $\hat A,\hat B$ on a finite-dimensional Hilbert space the Hilbert-Schmidt inner product is $\langle\hat A,\hat B\rangle \coloneq {\rm Tr}[{\hat A\dg \hat B}]$). This choice is consistent with the geometry underlying the norm used to quantify propagation. The Hermitian component $\mathcal L_{\rm d}$ governs the rate of change of the Hilbert–Schmidt norm via
\begin{equation} \label{HS:norm:change}
	\frac{d}{dt} \big\Vert \rho(t)\big\Vert_{\rm HS}^2 = 2 \langle \rho, {\cal L}_{\rm d}(\rho) \rangle.
\end{equation}
While ${\cal L}_{\rm d}$ is not necessarily negative semidefinite, it encodes the magnitude and direction of dissipative effects, including both contraction and expansion along different operator directions. In contrast, the anti-Hermitian component ${\cal L}_{\rm nd}$ preserves the norm and generates rotations in operator space (nondissipative)\footnote{It can be shown that $\langle \rho, {\cal L}_{\rm nd}(\rho) \rangle$ is purely imaginary and changes only the phase or direction in operator space, not the length. Thus it does not contribute to the Hilbert–Schmidt norm change. See \arf{appn:ah:Lind} for the details of calculations.}. Equation (\ref{HS:norm:change}) indicates that $\mathcal L_{\rm d}$ governs state norm growth and/or decay. 

We quantify these features by introducing two scalar metrics: the \emph{dissipative strength}
\begin{equation} \label{dis:strength} 
	\delta(\mathcal L) := \|\mathcal L_{\rm d}\| = \max_{\lambda \in \mathrm{spec}(\mathcal L_{\rm d})} |\lambda|,
\end{equation}
which measures overall irreversibility\footnote{Entropy production rate when defined with respect to a steady state, is likewise controlled by the dissipative component of the generator. In this sense, $\delta(\mathcal L)$ can be viewed as a generator-level proxy for irreversibility that does not rely on steady-state structure.}. Here and throughout $\|\cdot\|$ denotes the operator norm induced by the Hilbert--Schmidt inner product: 
\begin{equation} \label{eqn:OptNorm:HS}
	\|\mathcal L\| := \sup_{X \neq 0} \frac{\|\mathcal L(X)\|_{\rm HS}}{\|X\|_{\rm HS}}.
\end{equation}
 Since $\mathcal L_{\rm d}$ is Hermitian with respect to the Hilbert--Schmidt inner product, its induced operator norm equals the spectral radius (\arf{app:hermitian_norm}). Thus, it captures the worst-case rate of contraction or expansion over all directions in operator space. Importantly, it quantifies only the magnitude of dissipative action and does not by itself imply monotonic decay. The second measure is the \emph{nonnormality}~\cite{TreEmb05}
\begin{equation} \label{nn:gen} 
	\eta(\mathcal L) := \left\|\left[\mathcal L, \mathcal L^\dagger\right]\right\|,
\end{equation}
which quantifies the extent to which the generator induces directional mixing in operator space, and vanishes exactly for normal generators.  Importantly, it provides a computationally inexpensive measure that correlates with the onset of pseudospectral effects. 

The quantities $\delta(\mathcal L)$ and $\eta(\mathcal L)$ are not independent. Using the commutator bound, one finds\footnote{We note that $\delta(\mathcal L)$ and $\eta(\mathcal L)$ do not define a closed parameter space, as the accessible region depends additionally on the scale of the anti-Hermitian component $\|\mathcal L_{\rm nd}\|$. Consequently, the  regimes introduced here should be understood as scaling regimes characterized by the relative magnitude of $\eta(\mathcal L)$ and $\delta(\mathcal L)^2$, rather than as sharply defined geometric regions in the $(\delta,\eta)$ plane..}
\begin{equation} \label{bound:on:eta}
	\eta(\mathcal L) \le  4\, \delta(\mathcal L)\, \|\mathcal L_{\rm nd}\|.
\end{equation}
In particular, this bound implies the following structural constraint.

\begin{proposition}[No directional flow without dissipation]
	If $\delta(\mathcal L)=0$, then $\eta(\mathcal L)=0$.
\end{proposition}

\begin{proof}
	This follows immediately from \erf{bound:on:eta}. Equivalently, $\delta(\mathcal L)=0$ implies $\mathcal L=\mathcal L_{\rm nd}$ is anti-Hermitian and hence normal, so $\eta(\mathcal L)=0$.
\end{proof}

Thus, nonnormal effects require a nonzero dissipative component of the generator: directional flow in operator space cannot arise without a nonzero Hermitian term of the generator. However, the converse is not true—dissipative generators may still be normal, and therefore normality does not imply reversible (unitary) dynamics.

  \begin{proposition}[Decoupling for normal generators]
 	Let $\mathcal L$ be a quantum dynamical generator with decomposition given in \erfa{L:d:nd}{Ld:Lnd}. If $\eta(\mathcal L)=0$, then the dynamical map factorizes as
 	\begin{equation*}
 		e^{t\mathcal L} = e^{t\mathcal L_{\rm d}}\; e^{t\mathcal L_{\rm nd}}.
 	\end{equation*}
 \end{proposition}
 
 \begin{proof}
 If $[\mathcal L,\mathcal L^\dagger]=0$, then $[\mathcal L_{\rm d},\mathcal L_{\rm nd}] = \tfrac{1}{2}[\mathcal L^\dagger,\mathcal L] = 0$. Therefore, the factorization of the exponential follows from the commutativity of $\mathcal L_{\rm d}$ and $\mathcal L_{\rm nd}$.
 \end{proof}
 \noindent
Thus, the nonnormal case leads to an intrinsic coupling between dissipation and directional flow. We now define dynamical classes based on the dissipation strength and nonnormality measure.

\begin{definition}[Dynamical classes]
	For $\delta \ge 0$ and $\eta \ge 0$, define
	\begin{equation}
		\mathcal C(\delta,\eta) := \{ \mathcal L \; | \; \delta(\mathcal L) \in [0,\delta],\; \eta(\mathcal L) \in [0,\eta] \}.
	\end{equation}
\end{definition}

This definition induces a natural inclusion structure: if $\delta_1 \le \delta_2$ and $\eta_1 \le \eta_2$, then $\mathcal C(\delta_1,\eta_1) \subseteq \mathcal C(\delta_2,\eta_2)$. In particular, three important regimes emerge:
\begin{equation}
	\mathcal C(0,0) \subsetneq \mathcal C(\delta,0) \subsetneq \mathcal C(\delta,\eta),
\end{equation}
for $\delta,\eta \in {\mathbb R}^+$. Here, $\mathcal C(0,0)$ corresponds to closed (purely Hamiltonian) dynamics, $\mathcal C(\delta,0)$ to normal Lindbladians, and $\mathcal C(\delta,\eta)$ to generic open-system generators. These inclusions are strict, reflecting the increasing structural complexity of the dynamics.

The quantities $\delta(\mathcal L)$ and $\eta(\mathcal L)$ depend on the choice of operator norm. In this work, we employ norms induced by the Hilbert--Schmidt inner product (equivalently, spectral norms of matrix representations of $\mathcal L$), which provide basis-independent measures and enable direct comparison between generators. While different norm choices may lead to quantitative variations, the inclusion structure and the implication $\delta=0 \Rightarrow \eta=0$ remain unchanged.

\subsection{Geometric classification and dynamical regimes}
The classification induced by the pair $(\delta(\mathcal L),\eta(\mathcal L))$ admits a natural geometric representation, shown in \frf{fig:delta_eta_plane}. This construction provides a useful framework for identifying distinct dynamical  scaling-based  regimes. 
From the structural constraint $\delta(\mathcal L)=0 \Rightarrow \eta(\mathcal L)=0$, it follows that the only admissible point on the line $\delta=0$ is the origin. 
For any $\delta>0$, however, the nonnormality $\eta(\mathcal L)$ can become large relative to $\delta^2$ within bounded operator norm scaling.
Thus, the set of admissible generators corresponds to the region $\{(\delta,\eta): \delta>0,\ \eta\ge 0\}$ together with the isolated point $(0,0)$. 

Within this representation, the origin $(0,0)$ corresponds to purely Hamiltonian dynamics (blue point), the line $\eta=0$ with $\delta>0$ corresponds to normal dissipative generators with no directional flow (purple line), and the interior region $\delta>0,\ \eta>0$ captures generic nonnormal open systems exhibiting directional flow in operator space (gray region). This visualization highlights that nonnormality is not bounded above by dissipative strength, but rather requires it: directional flow can only arise in the presence of dissipation, yet can become arbitrarily large even for weakly dissipative generators (which requires large $\|\mathcal L_{\rm nd}\|$ through \erf{bound:on:eta}).

\begin{figure}[t]
	\centering
	\begin{tikzpicture}[scale=1.1]
		
		\fill[gray!15] (0.05,0) rectangle (6,6);
		
		\draw[-{Latex}, thick] (0,0) -- (6,0) node[right] {$\bm{\delta(\mathcal L)}$};
		\draw[-{Latex}, dashed, very thick] (0,0) -- (0,6) node[above] {$\bm{\eta(\mathcal L)}$};
		
		\draw[red,very thick] (0.05,0) -- (5.8,0);
		
		\filldraw[blue] (0,0) circle (2pt) node[below left, blue] {$\bm{(0,0)}$};;
		
		\fill[teal!30, opacity=0.4] (4.95,0.68) ellipse (0.7 and 0.5);
		\node[teal!50!black, align=center] at (4.95,0.7) {\textbf{Weakly}\\\textbf{nonnormal}};
		
		\draw[violet, thick, dotted] plot[domain=0.05:2.45] (\x, {\x*\x});
		\node[violet, rotate=70] at (1.5,1.4) {\textbf{Crossover}};
		\node[violet, rotate=75] at (2.55,5.4) {$\eta = \delta^2$};
		
		\draw[olive, thick, dash dot] (0,0) -- (1.5,6.01);
		\node[olive, rotate=75] at (1.05,3.0) {$\eta \le 4\,\delta\,\|\mathcal L_{\rm nd}\|$};
		
		\fill[brown!30, opacity=0.6] (1.2,5.3) ellipse (0.7 and 0.5);
		\node[brown!50!black, align=center] at (1.22,5.32) {\textbf{Strongly}\\\textbf{nonnormal}};
		
		\node[align=center] at (3.5,3.0) {\textbf{Generic open}\\\textbf{quantum systems}};
		
		\node[below, red] at (3,0) {\textbf{Normal dissipative}};
		
		\node[above right, blue, rotate=90] at (-0.05,0) {\textbf{Hamiltonian}};
		
	\end{tikzpicture}
	\caption{
	Geometric representation of quantum dynamical generators in the $(\delta,\eta)$ plane. 
	The origin $(0,0)$ corresponds to purely Hamiltonian dynamics (blue point), while the line $\eta=0$ with $\delta>0$ represents normal dissipative generators ( red  line). 
	The shaded region $\delta>0,\ \eta>0$ captures generic nonnormal open systems (gray area).
	The dimensionless ratio $\kappa(\mathcal L)=\eta(\mathcal L)/\delta(\mathcal L)^2$, \erf{eqn:eta-delta:ratio}, distinguishes different dynamical regimes. 
	The dotted violet curve indicates the crossover regime $\kappa(\mathcal L)\sim 1$. 
	The highlighted ovals illustrate the asymptotic regimes: weakly nonnormal systems ($\kappa(\mathcal L)\ll 1$) are shown in teal (bottom-right), where dissipation dominates, while strongly nonnormal systems ($\kappa(\mathcal L)\gg 1$) are shown in brown (top-left), where nonnormal effects dominate. 
	These regions are indicative rather than sharply defined boundaries. In addition, the dotted dashed olive line indicates a representative upper bound, \erf{bound:on:eta}, illustrating that nonnormality is structurally constrained and cannot be increased independently of the dissipative and anti-Hermitian components.
	}
	\label{fig:delta_eta_plane}
\end{figure}

While the region $\delta>0,\ \eta>0$ encompasses generic open-system dynamics, it does not admit a complete classification in terms of these scalar quantities alone. Nevertheless, meaningful  scaling  regimes can be identified based on the relative scaling of dissipative strength and nonnormality, \erf{eqn:eta-delta:ratio}. Analyzing these regimes provides useful insight into the resulting dynamics.
In particular, in the weakly nonnormal regime (teal oval), for which $\eta \ll \delta^2$, the generator is close to normal and transient amplification remains perturbatively small, so that the propagator norm is well approximated by purely exponential behavior governed by $\delta(\mathcal L)$. 
In contrast, in the strongly nonnormal regime (brown oval), for which $\eta \gg \delta^2$, nonnormal effects become significant and can induce substantial transient amplification, increasing sensitivity to perturbations and numerical error. 
The crossover regime $\eta \sim \delta^2$ marks the transition to this behavior and is depicted by the dotted violet curve in \frf{fig:delta_eta_plane}. The dash-dotted olive line depicts a representative upper bound, \erf{bound:on:eta}, showing that increasing nonnormality necessarily entails increasing either dissipation or the magnitude of the anti-Hermitian component.

This framework establishes a hierarchy of quantum dynamical generators based on dissipative strength and nonnormality, providing a unified, generator-level characterization of irreversibility and directional flow that will be used throughout the remainder of this work.

\subsection{Operational interpretation of metrics and connection to simulation complexity} \label{sec:formalism:metrics}

The quantities $\delta(\mathcal L)$ and $\eta(\mathcal L)$ introduced above admit a direct operational interpretation in terms of the propagator ${\rm exp}({t\mathcal L})$ and the stability of quantum evolution.
We first consider the role of the dissipative component $\mathcal L_{\rm d}$. From \erf{HS:norm:change}, one obtains the bound (consult \arf{app:propagation_bounds} for the details of derivations) 
\begin{equation}
	\frac{d}{dt} \|\rho(t)\|_{\rm HS}^2 
	\le 2\delta(\mathcal L)\, \|\rho(t)\|_{\rm HS}^2.
\end{equation}
Applying Gr\"{o}nwall’s inequality yields~\cite{Gronwall1919}
\begin{equation} \label{eq:prop_bound_main}
	\|\rho(t)\|_{\rm HS} \le e^{t\delta(\mathcal L)} \|\rho(0)\|_{\rm HS}.
\end{equation}
Thus, $\delta(\mathcal L)$ sets the maximal exponential rate of norm change and defines an intrinsic dynamical timescale $\tau \sim 1/\delta(\mathcal L)$. In particular, it controls the overall magnitude of evolution over time $t$.

For normal generators, $\eta(\mathcal L)=0$,
the propagator factorizes as $e^{t\mathcal L} = e^{t\mathcal L_{\rm d}} e^{t\mathcal L_{\rm nd}}$. Since $\mathcal L_{\rm nd}$ is anti-Hermitian, $e^{t\mathcal L_{\rm nd}}$ is norm-preserving, and one obtains
\begin{equation} \label{eqn:map:norm:normal}
	\left\|e^{t\mathcal L}\right\| = \left\|e^{t\mathcal L_{\rm d}}\right\| \le e^{t\delta(\mathcal L)}.
\end{equation}
 The propagator norm $\left\|e^{t\mathcal L}\right\|$ quantifies the maximum possible amplification of perturbations under the dynamics\footnote{If the initial state is perturbed $\rho(0) \to \rho(0) + \tilde\delta \rho(0)$, then $\left\|\tilde\delta \rho(t) \right\| \le 	\left\|e^{t\mathcal L}\right\| 	\left\|\tilde\delta \rho(0)\right\|$.  } and therefore characterizes the stability of quantum evolution in operator space. In this case, the dynamics is fully controlled by spectral properties, and no transient amplification beyond exponential scaling, set by $\delta(\mathcal L)$, occurs.

In contrast, for nonnormal generators with $\eta(\mathcal L) > 0$, the propagator may exhibit transient amplification that is not captured by the spectrum alone. To isolate this effect, we factor out the baseline exponential scaling (set by the dissipative component) and write
\begin{equation} \label{eqn:norm:amp}
	\left\|e^{t\mathcal L}\right\| = A(t)\, e^{t\delta(\mathcal L)},
\end{equation}
where $A(t)$ is a time-dependent amplification factor, with $A(t)=1$ in the normal case and we refer to $A(t) > 1$ as transient amplification. By construction, $A(t)$ captures deviations from the purely exponential behavior governed by $\delta(\mathcal L)$ and reflects the nonorthogonality of eigenmodes of the generator and the resulting directional flow in operator space (see \arf{app:propagation_bounds} for further discussion). 

These bounds have direct implications for the propagation of numerical errors. If an approximation to the propagator incurs an error of size $\epsilon$, then the resulting state error satisfies:
\begin{equation} \label{eqn:state:error:gen}
	\|\Delta \rho(t)\| \;\lesssim\; \epsilon\,\left\|e^{t\mathcal L}\right\|
	\; = \; \epsilon\,A(t)\, e^{t\delta(\mathcal L)}.
\end{equation}
Thus, nonnormality increases the effective precision requirement by a factor $A(t)$, amplifying errors beyond what is expected from dissipative scaling alone.

From the perspective of simulation complexity, this leads to a clear separation of roles. The dissipative strength $\delta(\mathcal L)$ sets the intrinsic timescale of evolution and determines the baseline exponential scaling of both the propagator norm and numerical error amplification over the simulation time. In contrast, the nonnormality $\eta(\mathcal L)$ controls the deviation from normal-mode propagation and acts as a condition-number–like enhancement of perturbation growth beyond the baseline set by $\delta$, and thus governs the sensitivity of the dynamics to finite precision $\epsilon$ beyond this baseline. In particular, generators with $\eta(\mathcal L)=0$ exhibit stable evolution with no transient amplification, while for $\eta(\mathcal L)>0$, the growth of $A(t)$ increases simulation overhead associated with achieving a fixed accuracy.

\section{Classes of quantum dynamical systems} \label{sec:classes}

The framework introduced in Sec.~\ref{sec:formalism} provides a natural way to organize quantum dynamical generators according to their dissipative strength and degree of nonnormality. In this section, we illustrate this structure by examining representative dynamical classes and selected  illustrative  examples where appropriate.  Several of the examples discussed below can be viewed as limiting cases of a driven dissipative qubit with dephasing and relaxation, illustrating how different dynamical regimes emerge from the relative scaling of coherent and dissipative processes. 

While the classification defined by $(\delta,\eta)$ is complete at the level of these scalar quantities, the interior region $\delta>0,\ \eta>0$ encompasses a wide variety of qualitatively distinct dynamics and does not admit a full classification in terms of these quantities alone. 
Nevertheless, the extremal and boundary cases already capture important physical regimes and provide useful intuition for the role of dissipation and nonnormality.

We therefore consider three representative classes: 
(i) purely Hamiltonian systems with $\delta=0$ and $\eta=0$, 
(ii) normal dissipative systems with $\delta>0$ and $\eta=0$, and 
(iii) generic nonnormal open systems with $\delta>0$ and $\eta>0$. 
For each class, we analyze the structure of the generator and discuss its implications for dynamical behavior and stability.

Before proceeding, we briefly comment on the computational complexity of simulating these classes of dynamics. While we do not provide detailed derivations here, it is useful to note that for purely Hamiltonian dynamics ($\delta=\eta=0$), efficient simulation algorithms achieve linear scaling in evolution time with additive logarithmic dependence on precision~\cite{LowChu17}. When dissipation is present but the generator remains normal ($\delta > 0$, $\eta=0$), similar scaling behavior is expected~\cite{BorMar25}, albeit with coefficients reflecting the dissipative strength through norm bounds on the generator. In the more general nonnormal regime ($\delta > 0$, $\eta > 0$), transient amplification can increase simulation costs by coupling time and precision requirements. Further discussion can be found in Appendices~\ref{app:propagation_bounds} and \ref{app:weak_nonnormal_bound}.

\subsection{Hamiltonian dynamics: $\mathcal C(0,0)$}

The class $\mathcal C(0,0)$ corresponds to purely Hamiltonian dynamics, for which the generator takes the form
\begin{equation} \label{eqn:H:lindbladian}
	\mathcal L_{\rm H}(\rho) = -i[\hat H,\rho].
\end{equation}
In this case, the generator is anti-Hermitian with respect to the Hilbert--Schmidt inner product, and therefore $\mathcal L_{\rm d} = 0$ and $\mathcal L_{\rm nd} = \mathcal L_{\rm H}$. It follows immediately that
\begin{equation} \label{constraint:H}
	\delta(\mathcal L_{\rm H}) = 0, \qquad \eta(\mathcal L_{\rm H}) = 0,
\end{equation}
so that closed quantum systems occupy the isolated point at the origin, $(\delta,\eta) = (0,0)$ in Fig.~\ref{fig:delta_eta_plane}.
The corresponding propagator,
 $e^{t\mathcal L_{\rm H}}(\rho) = e^{-i\hat H t}\, \rho\, e^{i\hat H t}$, preserves the Hilbert--Schmidt norm:
\begin{equation}
	\|e^{t\mathcal L_{\rm H}}(\rho)\|_{\rm HS} = \|\rho\|_{\rm HS}.
\end{equation}
Equivalently, the induced operator norm satisfies $	\left\|e^{t\mathcal L_{\rm H}}\right\| = 1$, indicating that no amplification or contraction of perturbations occurs ($A(t)=1$ in \erf{eqn:norm:amp}). Therefore, such dynamics can be viewed as generating coherent transport on the space of operators, without any dissipative attenuation.

From the perspective of simulation complexity, closed quantum systems are characterized by \erfs{eqn:H:lindbladian}{constraint:H},  and preserve the dynamical map's norm. Therefore the simulation cost takes the form 
\begin{equation}
\mathrm{cost} = O \left[ t \left\|\mathcal{L}_{\rm H}\right\| + \log\!\left({1}/{\epsilon^*}\right) \right],
\end{equation}
where $\epsilon^*$ denotes the target simulation accuracy~\cite{BerSan07,ChiDam10}, and the first term in square brackets is the effective dimensionless evolution time. For the Hamiltonian evolution one has $\|\mathcal{L}_{\rm H}\| = 2 \|\hat H\|$, and the complexity reduces to the standard linear-in-time scaling with additive logarithmic dependence on precision. For $\|\hat H\|= O(1)$, this yields the optimal scaling $O(t + \log(1/\epsilon^*))$, consistent with established Hamiltonian simulation algorithms~\cite{LowChu17, ChiSu18}. This class therefore provides a natural reference point for assessing the impact of dissipation and nonnormality in more general quantum dynamical systems.

A representative example is provided by the Jaynes--Cummings model, describing the coherent interaction between a two-level system and a single bosonic mode. 
In the absence of dissipation, the dynamics is fully unitary and exhibits coherent oscillations between excitation subspaces. Within the $(\delta,\eta)$ framework, the dynamics remains confined to the origin regardless of system parameters, reflecting the absence of both dissipation and nonnormal effects.

\subsection{Normal dissipative dynamics: $\mathcal C(\delta,0)$}

We now consider dissipative generators that remain \emph{normal}, i.e., $\eta(\mathcal L_{\rm n})=0$ while $\delta(\mathcal L_{\rm n})>0$. 
As discussed in Sec.~\ref{sec:formalism:metrics}, the anti-Hermitian component generates norm-preserving evolution, while the Hermitian component governs contraction or growth.
Consequently, the propagator norm is determined entirely by the dissipative strength, as in \erf{eqn:map:norm:normal}.

Our goal is not to describe the most general form of each class, but rather to isolate the structural features captured by the quantities $\delta(\mathcal L_{\rm n})$ and $\eta(\mathcal L_{\rm n})$.
To this end, here we restrict to purely dissipative normal generators without Hamiltonian terms for simplicity. This simplification allows us to directly expose the role of the dissipative component $\mathcal L_{\rm d}$ in setting the scale of evolution. We emphasize, however, that while the Hamiltonian contribution does not affect the dissipative strength $\delta(\mathcal L)$, it contributes to the anti-Hermitian component $\mathcal L_{\rm nd}$ and can induce nonnormality through the commutator $[\mathcal L_{\rm d},\mathcal L_{\rm nd}]$, generically leading to $\eta(\mathcal L)>0$ \footnote{Importantly, the presence of a Hamiltonian contribution does not, by itself, imply nonnormality. Normal generators may still include coherent dynamics, provided the anti-Hermitian and dissipative components commute, so that $\eta(\mathcal L)=0$.}.

For normal generators, the evaluation of $\delta(\mathcal L_{\rm n})$ simplifies significantly. Since $\mathcal L_{\rm n}$ is diagonalizable in an orthonormal operator basis, the Hermitian component $\mathcal L_{\rm d}$ has eigenvalues given by the real parts of the eigenvalues of $\mathcal L_{\rm n}$. As a result,
\begin{equation} \label{eqn:dis:str:normal}
	\delta(\mathcal L_{\rm n}) = \max_\beta |\mathrm{Re}(\lambda_\beta)|, \quad {\rm for}\; \eta(\mathcal L_{\rm n})=0,
\end{equation}
where $\{\lambda_\beta\}$ are the eigenvalues of $\mathcal L_{\rm n}$. Thus, for normal dissipative dynamics, the dissipative strength is completely determined by the spectrum of the generator.

Therefore, in contrast to the nonnormal case, no transient amplification occurs and the dynamics is entirely characterized by exponential scaling set by $\delta(\mathcal L_{\rm n})$. This makes normal dissipative generators a particularly tractable class, both analytically and from the perspective of quantum simulation. In what follows, we consider two examples.

\subsubsection{Pure dephasing} \label{sec:pure:deph}
Consider a single qubit undergoing pure dephasing with Lindblad operator $\hat L_{\rm z} = \sqrt{\gamma_{\rm z}}\, \hat \sigma_{\rm z}$ such that:
\begin{equation} \label{diss:pure:deph}
	\mathcal L_{\rm z}(\rho) = \gamma_{\rm z} \left( \hat \sigma_{\rm z} \rho \hat \sigma_{\rm z} - \rho \right),
\end{equation}
where $\gamma_{\rm z} \in {\mathbb R}^+$ is the dephasing rate, and we assume the Hamiltonian $\hat H=0$. The Pauli operators are denoted by $(\hat \sigma_{\rm x}, \hat \sigma_{\rm y}, \hat \sigma_{\rm z})$. This generator is Hermitian in the Hilbert--Schmidt inner product, so that
\begin{equation}
	\mathcal L_{\rm d} = \mathcal L_{\rm z}, \qquad \mathcal L_{\rm nd}=0.
\end{equation}
 It follows immediately that\footnote{In the Pauli operator basis, the generator acts diagonally with eigenvalues 
 	$\{0,-2\gamma_{\rm z},-2\gamma_{\rm z},0\}$, so that $\|\mathcal L_{\rm z}\|$ (the largest singular value) is $2\gamma_{\rm z}$.}
\begin{equation}
	\delta(\mathcal L_{\rm z}) = 2\gamma_{\rm z}, 
	\qquad
	\eta(\mathcal L_{\rm z}) = 0.
\end{equation}
Thus, the Lindblad structure produces a doubling of the decay rate at the superoperator level. In other words, $\delta(\mathcal L_{\rm z})$ is not just a physical rate--it is a worst-case  operator-space rate.

Physically, the dynamics suppresses coherences in the computational basis while leaving populations invariant. In operator space, this corresponds to anisotropic contraction along off-diagonal directions, with no rotational component. Since the generator is normal, the evolution exhibits no transient amplification. Consequently, both the propagator norm and the induced evolution of states are fully controlled by the exponential scaling set by the dissipation strength, as reflected in~\erfa{eq:prop_bound_main}{eqn:map:norm:normal}.

This example extends straightforwardly to $K$-qubit systems with independent dephasing channels, where $\delta(\mathcal L^{K}_{\rm z}) = 2 \sum_{k=1}^K \gamma_{{\rm z}_k}$ scales additively with the number of qubits, while $\eta(\mathcal L^{K}_{\rm z})$ remains zero (See \arf{app:multi_qubit_dephasing} for the details of calculations). This demonstrates that independent dephasing provides a family of normal Lindbladians with extensive dissipative strength but no transient amplification.

\subsubsection{Structured dissipators}
A broader class of normal dissipative generators arises when the Lindblad operators satisfy
\begin{equation} \label{eqn:sd:condition}
	\sum_j \hat L_j^\dagger \hat L_j = \Gamma \, \hat I,
\end{equation}
for some $\Gamma \in {\mathbb R}^+$. In this case, the dissipative part of the generator simplifies to ($\hat H=0$)
\begin{equation}
	\mathcal L_{\rm s}(\rho) = \sum_j L_j \rho L_j^\dagger - \Gamma \rho.
\end{equation}
Such generators often arise in Pauli channels and in constructions used for dissipative state preparation. Under suitable symmetry conditions (e.g., mutually orthogonal jump operators forming a unitary operator basis), the generator is normal, $\eta(\mathcal L)=0$.
In these cases, the dissipative strength is set by the overall rate $\Gamma$, while the absence of nonnormality ensures that \erf{eqn:map:norm:normal} is satisfied with no additional amplification factor.

When the Lindblad operators obey the constraint in \erf{eqn:sd:condition}, the generator takes the form $\mathcal L_{\rm s} = \mathcal J - \Gamma\,\mathcal{I}$, where $\mathcal J(\rho) \coloneq \sum_j L_j \rho L_j^\dagger$ is a completely positive map (also known as the quantum jump map), and $\mathcal I$ is the identity superoperator. In this case, the spectrum $\{\lambda_\beta\}$ of $\mathcal L_{\rm s}$ is given by a uniform shift of the spectrum $\{\Lambda_\beta\}$ of $\mathcal J$. That is, $\lambda_\beta = \Lambda_\beta - \Gamma$, and the dissipation strength is obtained through \erf{eqn:dis:str:normal}. While this condition imposes strong structural constraints, it does not uniquely determine the spectrum of $\mathcal L_{\rm s}$ without additional assumptions on the operators $\{L_j\}$.

The spectrum becomes analytically tractable when the map $\mathcal J(\rho)$ is diagonal in a simple orthonormal operator basis. This occurs, for example, in Pauli channels, where conjugation by Pauli operators preserves the Pauli basis up to a sign, making all Pauli operators eigenoperators of $\mathcal J$, despite the fact that the underlying operators do not commute. In \arf{appn:1qb:1pauli} we consider a single qubit subject to the general Pauli channel.

The above structure has direct implications for simulation complexity. Since $\eta(\mathcal L_{\rm n})=0$, the propagator norm obeys \erf{eqn:map:norm:normal}, and numerical errors are amplified only by the baseline exponential factor set by the dissipative strength\footnote{Normal dissipative generators define a regime in which dissipation and unitary-like rotations are decoupled, and no transient amplification beyond exponential scaling occurs. Consequently, the amplification factor satisfies $A(t)=1$ for all $t$.}. 
For a target accuracy $\epsilon^*$, this implies that the required precision per step scales as $\epsilon \sim \epsilon^* e^{-t\delta(\mathcal L_{\rm n})}$ via \erf{eqn:state:error:gen}. Consequently, Using known efficient Hamiltonian- and Lindbladian-simulation algorithms~\cite{BorMar25}, the overall computational cost scales as
\begin{equation} \label{eqn:cost:normal}
	\mathrm{cost} = {O} \left[
	t\,\delta(\mathcal L_{\rm n})
	+ \log\!\left({1}/{\epsilon^*}\right)
	\right],
\end{equation}
 exhibiting additive dependence on evolution time and target precision, without additional overhead from transient effects (where optimal simulation algorithms achieve linear scaling in $t$, while more general methods may incur higher polynomial overhead). In this regime, the induced operator norm of the generator coincides with that of its Hermitian part, $\|\mathcal L_{\rm n}\| = \|\mathcal L_{\rm d}\| = \delta(\mathcal L)$. This regime therefore provides a natural benchmark for efficient simulation and a baseline against which the impact of nonnormality can be quantified. In contrast, nonnormal generators introduce transient amplification through $A(t)$, which effectively couples dissipative and  nonnormal  effects in the error propagation bound.

 \subsection{Nonnormal dissipative dynamics: $\mathcal C(\delta>0,\eta>0)$} \label{sec:nonnormal}
 
 We now turn to the generic class of open quantum systems for which both dissipative strength and nonnormality are nonzero. In contrast to the \emph{normal} case, the generator is no longer diagonalizable in an orthonormal operator basis, and the interplay between its Hermitian and anti-Hermitian components can lead to qualitatively richer dynamics, including transient amplification.

 To analyze this behavior, we recall \erf{L:d:nd} and the fact that $[\mathcal L_{\rm d},\mathcal L_{\rm nd}] \neq 0$. The deviation from normality is then encoded in nested commutators that enter the expansion of the propagator. A convenient way to formalize this is to use the interaction picture with respect to the dissipative part which yields the identity (see \arf{app:weak_nonnormal_bound} for details)
 \begin{equation}
 	e^{t\mathcal L} = e^{t\mathcal L_{\rm d}} \, \mathfrak T\left[ \exp\!\left( \int_0^t ds \, \widetilde{\mathcal L}_{\rm nd}(s) \right)\right], 
 \end{equation}
 where $\widetilde{\mathcal L}_{\rm nd}(s) = e^{-s\mathcal L_{\rm d}} \mathcal L_{\rm nd}\, e^{s\mathcal L_{\rm d}}$, and $\mathfrak T$ denotes the time-ordering symbol which orders superoperators according to their time arguments. The time dependence of $\widetilde{\mathcal L}_{\rm nd}(s)$ is governed by the commutator $[\mathcal L_{\rm d},\mathcal L_{\rm nd}]$, and admits an expansion in nested commutators. To leading order,
 \begin{equation}
 	\widetilde{\mathcal L}_{\rm nd}(s)
 	= \mathcal L_{\rm nd} + s [\mathcal L_{\rm nd}, \mathcal L_{\rm d}] + O(s^2).
 \end{equation}
Taking norms and using the definition of $\eta(\mathcal L)$ together with standard submultiplicativity arguments, one obtains (see Appendix~\ref{app:weak_nonnormal_bound} for the full bound)
 \begin{equation} \label{prop:norm:gen}
 	\left\| e^{t\mathcal L} \right\|
 	\;\le\;
 	e^{t\delta(\mathcal L)} \, \exp\!\big(
 	 \tfrac{1}{4} t^2 \eta(\mathcal L)
 	+ \cdots
 	\big),
 \end{equation}
 where the ellipsis denotes higher-order contributions arising from nested commutators. Here
 the omitted linear contribution $t\|\mathcal L_{\rm nd}\|$ arises from intermediate submultiplicativity estimates and does not contribute to intrinsic amplification, since $e^{t\mathcal L_{\rm nd}}$ is norm-preserving in the Hilbert--Schmidt inner product (see Appendix~\ref{app:weak_nonnormal_bound} for a refined treatment). In contrast, the quadratic and higher-order terms encode genuine nonnormal effects associated with commutator growth between $\mathcal L_{\rm d}$ and $\mathcal L_{\rm nd}$.
 
 While these contributions are higher-order in $`t$' for short times, they do not form a parametrically suppressed hierarchy at timescales $t \sim 1/\delta(\mathcal L)$. Instead, the resulting corrections can be organized in terms of the dimensionless ratio\footnote{While the classification is expressed in terms of $\eta(\mathcal L)$ and $\delta(\mathcal L)$, the transition between nonnormal regimes is governed by the dimensionless ratio. The  regimes should therefore be interpreted as scaling regimes determined by the relative magnitude of  $\eta$ and $\delta^2$, rather than  as  sharply separated regions in the $(\delta,\eta)$ plane.}(see \arf{app:nonnormal_regimes} for a derivation)
 \begin{equation}  \label{eqn:eta-delta:ratio}
 	\kappa(\mathcal L)\equiv \frac{\eta(\mathcal L)}{[\delta(\mathcal L)]^2},
 \end{equation}
 so that deviations from purely exponential behavior are controlled by this dimensionless ratio\footnote{Across all regimes, $\kappa(\mathcal L)$ serves as the organizing parameter controlling the departure from purely spectral dynamics toward transient amplification governed by nonorthogonality.}.
 This quantity is meaningful as a measure of nonnormality only when $\delta(\mathcal L)$ is finite. In the limit $\delta(\mathcal L)\to 0$, the generator becomes normal, and the apparent divergence of $\kappa(\mathcal L)$ reflects a singular normalization rather than genuine nonnormal behavior.
 Combining these observations, one obtains
 \begin{equation}
 	\|e^{t\mathcal L}\|
 	\le
 	e^{t \delta(\mathcal L)}
 	\exp\left( O[\kappa(\mathcal L)] \right),
 \end{equation}
 so that $\kappa(\mathcal L)$ provides a dimensionless measure of nonnormality relative to dissipative decay, and naturally organizes the dynamical behavior. In terms of the amplification factor introduced in Sec.~\ref{sec:formalism:metrics}, this corresponds to $A(t) \equiv \exp\left( O[\eta t^2] \right)$.
 
 In what follows, we distinguish three representative regimes.

\subsubsection{Weakly nonnormal regime}

When $\kappa(\mathcal L) \ll 1$, nonnormal effects are perturbative relative to dissipative decay, and the dynamics is well approximated by that of a normal generator with the same dissipative strength. In this regime, transient amplification remains weak and is controlled by the dimensionless ratio $\kappa(\mathcal L)$.

More precisely, for timescales $t \sim 1/\delta(\mathcal L)$, the contribution of the time-ordered exponential remains close to norm-preserving, up to multiplicative corrections of order $\exp(O[\kappa(\mathcal L)])$. Consequently, the propagator norm deviates from purely exponential decay only by a small multiplicative factor, and transient amplification satisfies $A(t)=1+O(\kappa(\mathcal L)) \approx 1$.

This has direct implications for simulation complexity. The amplification of numerical errors differs from the normal case only by a constant factor independent of evolution time. As a result, weakly nonnormal generators inherit the same asymptotic complexity scaling as normal dissipative systems, Eq.~\ref{eqn:cost:normal}. Nonnormal effects therefore contribute only subleading corrections to prefactors, without altering the time–precision tradeoff, \ie without introducing of additional coupling between time and precision requirements).

\paragraph*{Weakly driven dephasing.}
A simple realization of the weakly nonnormal regime is provided by a single qubit subject to dephasing along $\hat\sigma_{\rm z}$ and a weak coherent drive along $\hat\sigma_{\rm x}$. The master equation obeys \erf{eqn:GKSL} where the Lindblad operator is $\hat L_{\rm z} = \sqrt{\gamma_{\rm z}}\, \hat \sigma_{\rm z}$, and the Hamiltonian reads $\hat H = \tfrac{\Omega}{2}\,\hat\sigma_{\rm x}$.
The dissipative strength is $\delta({\mathcal L}_{{\rm z}_{\rm H}}) = 2\gamma_{\rm z}$, while the nonnormality is determined by the commutator $[\mathcal L_{\rm d},\mathcal L_{\rm nd}]$. A direct calculation yields 
$\eta({\mathcal L}_{{\rm z}_{\rm H}}) = 2\|[\mathcal L_{\rm d},\mathcal L_{\rm nd}]\| = O(\gamma_{\rm z}\Omega)$.
It follows that $\kappa(\mathcal L_{{\rm z}_{\rm H}}) = O\!\left({\Omega}/{\gamma_{\rm z}}\right)$,
so that for $\Omega \ll \gamma_{\rm z}$ the generator lies in the weakly nonnormal regime.
In this case, the dynamics is dominated by exponential decay of coherences, with only small transient distortions induced by the noncommuting Hamiltonian contribution.

\subsubsection{Crossover regime}

When $\kappa(\mathcal L) \sim 1$, nonnormal effects become comparable to dissipative decay, and the perturbative description of the weakly nonnormal regime breaks down. As discussed above, the deviation from purely exponential behavior is controlled by the ratio in \erf{eqn:eta-delta:ratio}, which is now $O(1)$.
Therefore, one obtains
\begin{equation} \label{gen:norm:crossover}
	\left\|e^{t\mathcal L}\right\|
	\le
	e^{t\delta(\mathcal L)} \exp\big(O(1)\big),
\end{equation}
so that the amplification factor remains finite and satisfies $A(t) = O(1)$, with a prefactor that is not parametrically small\footnote{The amplification factor satisfies $A(t) \le \exp(O(\kappa(\mathcal L)))$, which remains $O(1)$ in the crossover regime.}. Thus, the contribution of the time-ordered exponential to the propagator norm is no longer perturbative, but remain bounded. In this regime, the propagator may exhibit order-one transient amplification at intermediate times before eventual decay sets in. While such amplification is not parametrically large, it reflects the breakdown of purely spectral control and can lead to nonmonotonic evolution in operator norms.

Physically, this behavior reflects the competition between noncommuting components of the generator. While $\mathcal L_{\rm d}$ enforces contraction along certain directions in operator space, the nonnormal component $\mathcal L_{\rm nd}$ can rotate states into directions that are less  damped, leading to temporary growth in norm. In contrast to the weakly nonnormal regime, these effects are not parametrically suppressed.

The crossover regime also has important implications for simulation complexity. Since $A(t)$ remains bounded by a constant independent of evolution time, the propagator norm retains the form in \erf{gen:norm:crossover}. As a result,  f or a target accuracy $\epsilon^*$, the required precision per step satisfies
\begin{equation}
	\epsilon \sim \epsilon^* \, e^{-O(1)}.
\end{equation} 
Thus, since $A(t) = O(1)$, the per-step precision only requires an $O(1)$ overhead relative to the target global error $\epsilon^*$. Unlike the strongly nonnormal regime, there is no parametrically large amplification, and the simulation cost remains essentially controlled by the dissipative strength $\delta(\mathcal L)$.

Consequently, the overall simulation cost retains the same asymptotic scaling as in the normal case, \erf{eqn:cost:normal}, but with nonuniversal constant prefactors arising from transient amplification. This reflects the emergence of coupling between dissipative and nonnormal effects, even though no parametrically large overhead is incurred. In this sense, the crossover regime marks the transition between structurally stable dynamics (controlled entirely by $\delta$) and regimes where nonnormality qualitatively affects short-time behavior without yet altering asymptotic scaling.

\begin{figure}[t]
	\centering
	\includegraphics[width=0.99\linewidth]{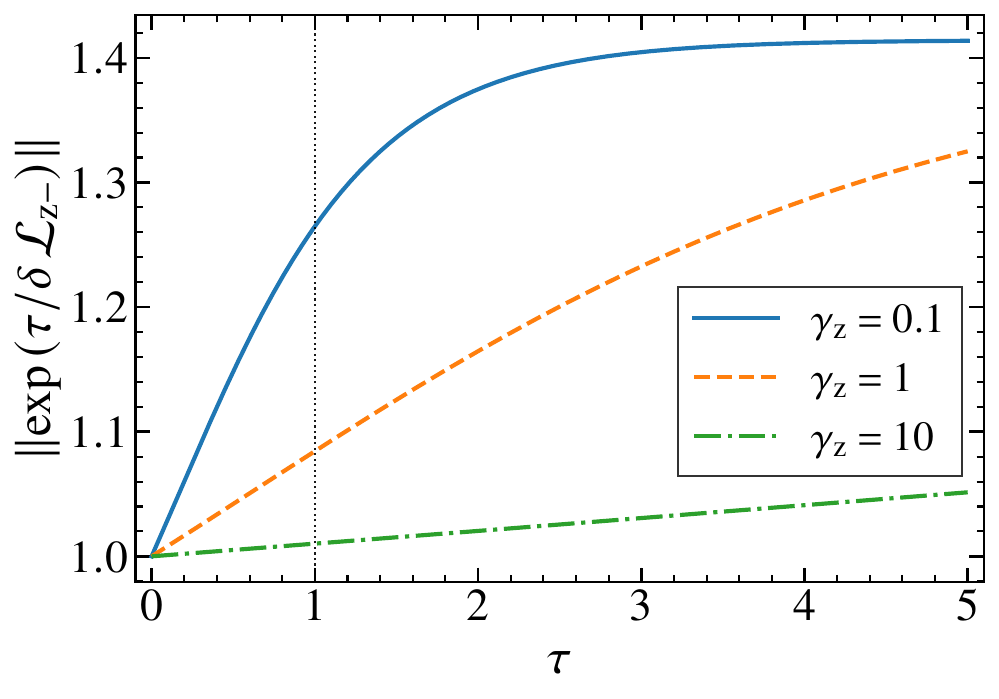}
	\caption{ Propagator norm $\left\|\exp\left[{\frac{\tau}{\delta(\mathcal{L}_{\rm z-})} \,\mathcal{L}_{\rm z-}}\right]\right\|$ for a single-qubit Lindbladian with
		competing dephasing ($\gamma_{\rm z}$) and amplitude damping ($\gamma_-$).
		The solid blue, dashed orange, and dash-dotted green curves correspond to
		$\gamma_{\rm z}=0.1, 1$, and $10$, respectively, with
		$\gamma_- = 1$ fixed. The vertical dotted line at $\tau=1$ marks the dimensionless transient timescale corresponding to the characteristic times $t=1/\delta(\mathcal L_{\rm z-})$. Increasing $\gamma_{\rm z}$ enhances the dissipative strength while leaving the nonnormality fixed by $\gamma_-$, leading to a progressive suppression of transient amplification. The strongest deviation from purely dissipative behavior occurs when $\kappa(\mathcal L_{\rm z-})$ is of order unity. The intermediate case $\gamma_{\rm z} = \gamma_-$ lies between these two limits and exhibits correspondingly moderate transient effects.
	}
	\label{fig:dephasing_relaxation_norm}
\end{figure}

\paragraph*{Competing dissipative channels.}
A minimal example is provided by a single qubit subject to two noncommuting dissipative processes. The Lindbladian, \erf{eqn:GKSL}, is constructed by two jump operators: dephasing process is represented by $\hat L_{\rm z}$ and the Lindblad operator associated with relaxation phenomenon is $\hat L_{-} = \sqrt{\gamma_-}\,\hat \sigma_-$, where $\gamma_- \in \mathbb R^+$ is the relaxation rate, and $\sigma_\pm = (\sigma_{\rm x} \pm i\sigma_{\rm y})/2$.
The dephasing term is  normal in the Hilbert--Schmidt inner product, whereas the amplitude-damping channel is itself nonnormal. Consequently, the combined generator satisfies  $\eta(\mathcal L_{\rm z-}) \neq 0$.
A direct calculation in the Pauli operator basis shows that (\arf{app:dephasing_relaxation})
\begin{equation} \label{rates:dephase:relax}
	\delta(\mathcal L_{\rm z-})  = 2\gamma_{\rm z} + \frac{\gamma_-}{2}, 
	\qquad
	\eta(\mathcal L_{\rm z-})  = \sqrt{2}\, \gamma_-^2.
\end{equation}
The corresponding ratio $\kappa(\mathcal L_{\rm z-})$,is small when $\gamma_{\rm z}\gg\gamma_-$, corresponding to a weakly nonnormal regime. By contrast, for $\gamma_{\rm z}\lesssim\gamma_-$ one finds $\kappa(\mathcal L_{\rm z-})=O(1)$, placing the dynamics in the crossover regime\footnote{ 
The limiting cases recover familiar behavior. For $\gamma_-=0$ the dynamics reduces to pure dephasing and is normal, with $\eta(\mathcal L_{\rm z-})=0$. For $\gamma_{\rm z}=0$ one obtains pure amplitude damping, for which $\kappa(\mathcal L_{\rm z-})=\sqrt2$.}.

Figure~\ref{fig:dephasing_relaxation_norm} illustrates the propagator norm
$\left\|\exp\left[{\smallfrac{\tau}{\delta(\mathcal{L}_{\rm z-})} \,\mathcal{L}_{\rm z-}}\right]\right\|$
as a function of the dimensionless time
$\tau=t\,\delta(\mathcal L_{\rm z-})$.
Expressing the evolution in terms of $\tau$ removes the overall dissipative
timescale set by $\delta(\mathcal L_{\rm z-})$, allowing the influence of
nonnormality to be compared directly for different parameter choices.
The vertical dotted line at $\tau=1$, therefore marks one characteristic
dissipative time. It serves as a common reference point for comparing generators with different
$\delta(\mathcal L_{\rm z-})$, rather than indicating a sharp dynamical transition.

As the dephasing rate increases, the dissipative strength
$\delta(\mathcal L_{\rm z-})$ grows while the nonnormality
$\eta(\mathcal L_{\rm z-})$ remains fixed by the relaxation rate
$\gamma_-$. Consequently, the ratio
$\kappa(\mathcal L_{\rm z-})$ decreases, and the propagator norm
approaches the behavior expected for a normal generator. Thus,
$\gamma_{\rm z}=0.1$ (solid blue) exhibits the largest deviation from
purely dissipative evolution, whereas
$\gamma_{\rm z}=10$ (dash-dotted green) lies deep in the weakly nonnormal
regime and remains close to the normal limit. The intermediate case
$\gamma_{\rm z}=\gamma_-$ (dashed orange) interpolates smoothly between
these two behaviors.

The monotonic growth of the propagator norm should not be with
the noticeable transient overshoot often associated with strongly nonnormal
systems. Rather, for the weakly nonnormal and crossover regimes considered
here, nonnormality manifests itself through a gradual enhancement of the
propagator norm relative to the baseline exponential scaling set by
$\delta(\mathcal L_{\rm z-})$. In terms of the decomposition
$\|e^{t\mathcal L}\|=A(t)e^{t\delta(\mathcal L)}$, this corresponds to a
time-dependent amplification factor that modifies the evolution, according to \trf{tab:regimes}, without
producing a pronounced transient peak. Such peaks are expected only in the
strongly nonnormal regime, where $\kappa(\mathcal L)\gg1$, which is not realized in the present example.

The finite asymptotic value of the propagator norm reflects the balance between dissipation and nonnormality for the chosen generator. A larger plateau corresponds to a stronger cumulative enhancement of the propagator norm due to nonnormality, whereas increasing the dephasing rate lowers the plateau by suppressing this effect. 

\subsubsection{Strongly nonnormal regime}

When $\kappa(\mathcal L) \gg 1$, nonnormal effects dominate over dissipative decay at short and intermediate times. In this regime, the generator is far from normal, and the dynamics cannot be understood as a perturbation of purely dissipative evolution.  This regime is only accessible when the anti-Hermitian component of the generator dominates over dissipation, as required by the bound given in \erf{bound:on:eta}.

To characterize this behavior, we again consider the bound on the propagator norm, \erf{prop:norm:gen}.
However, in contrast to the weakly nonnormal regime, the higher-order contributions are now parametrically large, and the time-ordered exponential produces substantial corrections to the baseline exponential behavior.

As a result, the evolution can exhibit strong transient amplification, with $\|e^{t\mathcal L}\| \gg 1$ over a finite time window even when all eigenvalues of $\mathcal L$ have nonpositive real parts. This reflects the fact that the spectral properties of $\mathcal L$ alone no longer determine the dynamical behavior; instead, the nonorthogonality of eigenoperators and the associated pseudospectral structure become essential.

From the interaction-picture perspective, the rapid growth of nested commutators implies that $\widetilde{\mathcal L}_{\rm nd}(t)$ varies significantly over timescales $t \lesssim 1/\delta(\mathcal L)$, leading to nontrivial accumulation effects in the time-ordered exponential.  In operator-space terms, the evolution can repeatedly rotate operator components into directions that are only weakly damped, producing large temporary amplification before eventual decay. 

The strongly nonnormal regime also has significant implications for simulation complexity. As discussed above, the propagator norm can exhibit substantial transient amplification. Heuristically, one may expect
\begin{equation}
	\|e^{t\mathcal L}\|
	\;\lesssim\;
	\exp\!\left( c\,\kappa(\mathcal L) \right)
\end{equation}
at timescales $t \sim 1/\delta(\mathcal L)$, where $c = O(1)$ depends on details of the generator and the chosen norm. Such behavior is well known in the theory of nonnormal operators, where transient growth can be governed by pseudospectral properties rather than eigenvalues alone~\cite{TreEmb05}. As a result, numerical errors can be amplified by a factor that is not controlled solely by the dissipative strength.

\begin{table*}[t]
	\centering
	\begin{tabular}{c c c c c c c c l}
		\hline \hline
		\textbf{Class} & & $\bm{\delta(\mathcal L)}$ & & $\bm{\eta(\mathcal L)}$ & & \textbf{Dynamics} & & \textbf{Simulation Features}\Tstrut \\[4pt]
		\hline 
		Hamiltonian & & $0$ & & $0$ & & Unitary, norm-preserving & & $A(t)=1$, efficient scaling\Tstrut \\[4pt]
		Normal dissipative & & $>0$ & & $0$ & & Exponential decay, no directional flow & & $A(t)=1$, governed by $\delta(\mathcal L)$ \\[4pt]
		Weakly nonnormal &\hspace{10pt} & $>0$ &\hspace{10pt} & $\ll \delta^2$ &\hspace{10pt} & Near-exponential, small transient effects &\hspace{10pt} & $A(t) = 1 + O[\kappa(\mathcal L)]$ \\[4pt]
		Crossover & & $>0$ & & $\sim \delta^2$ & & competing decay/amplification & & $A(t)=O(1)$ \\[4pt]
		Strongly nonnormal & & $>0$ & & $\gg \delta^2$  & & Strong transient amplification & & $A(t) \sim \exp(O[\kappa(\mathcal L)])$ \\[4pt]
		\hline
		\hline
	\end{tabular}
	\caption{
		Summary of dynamical regimes classified by dissipative strength $\delta(\mathcal L)$ and nonnormality $\eta(\mathcal L)$. The dimensionless ratio $\kappa(\mathcal L)=\eta(\mathcal L)/[\delta(\mathcal L)]^2$ governs the transition between nonnormal regimes. The amplification factor $A(t)$, defined in \erf{eqn:norm:amp}, quantifies deviations from purely exponential behavior and satisfies $A(t)=1$ for normal generators ($\eta(\mathcal L)=0$).)
	}
	\label{tab:regimes}
\end{table*}

For a target accuracy $\epsilon^*$, this implies that the required precision per step scales as
\begin{equation}
	\epsilon \sim \epsilon^* \, e^{-O[\kappa(\mathcal L)]},
\end{equation}
reflecting the need to compensate for transient amplification. In other words, the per-step error must be exponentially smaller than the target global error due to intermediate growth in the propagator norm. Consequently, the overall simulation cost acquires an additional overhead that can be heuristically associated with $\kappa$,
\begin{equation}
	\mathrm{cost}
	= O\!\left[
	t\,\delta(\mathcal L)
	+\log\!\left({1}/{\epsilon^*}\right)
	\right] + O[\kappa(\mathcal L)],
\end{equation}
indicating that nonnormality introduces an additional overhead in the precision scaling proportional to $\kappa(\mathcal L)$. While the scaling with evolution time remains linear for optimal algorithms, the presence of strong nonnormality can significantly increase the overall computational cost, reflecting the intrinsic amplification of intermediate errors.

More broadly, this regime highlights the limitations of spectral intuition in open quantum systems. Even when all eigenvalues indicate decay, the transient dynamics can be highly nontrivial, underscoring the central role of nonnormality in determining physically relevant behavior.

Generally speaking, achieving $\kappa(\mathcal L) \gg 1$ requires a separation of scales in which the noncommutativity between $\mathcal L_{\rm d}$ and $\mathcal L_{\rm nd}$ grows parametrically faster than the dissipative strength itself. This is not typical in simple few-body Lindbladians with independent or weakly coupled dissipation channels, where one usually finds $\eta(\mathcal L) \lesssim \delta(\mathcal L)^2$ and hence $\kappa(\mathcal L) = O(1)$. Instead, strongly nonnormal behavior is expected to arise in more structured settings, such as systems with correlated dissipation, collective jump operators, or highly nonorthogonal operator bases. In this sense, while the strongly nonnormal regime is physically allowed, it is not generic in simple models, and its identification requires careful analysis of the interplay between dissipative structure and  coherent dynamics.
For example, strongly nonnormal behavior can arise in driven-dissipative systems when a coherent drive is sufficiently strong compared to the dissipative scale, effectively amplifying the noncommutativity between the Hermitian and anti-Hermitian components of the generator. A minimal realization of this mechanism is provided below.

\paragraph*{Strongly driven amplitude damping.} consider a single qubit subject to amplitude damping at rate $\gamma_-$ and a coherent resonant drive of strength $\Omega$. In this case, the Lindblad operator in \erf{eqn:GKSL} is given by $\hat L_- = \sqrt{\gamma_-}\, \hat \sigma_-$ and the Hamiltonian is denoted by $\hat H = \frac{\Omega}{2} \hat \sigma_{\rm x}$. In this case one finds $\delta(\mathcal L)\sim \gamma_-$ while the nonnormality scales as $\eta(\mathcal L)\sim \gamma_-^2 + \gamma_-\Omega$, leading to $\kappa(\mathcal L)\sim 1+\Omega/\gamma_-$. Thus, the strongly nonnormal regime $\kappa(\mathcal L)\gg 1$ is achieved in the strong-driving limit $\Omega\gg\gamma_-$. The derivation is provided in Appendix~\ref{app:driven_amplitude_damping}.

\section{conclusion} \label{sec:conclusion}
In summary, we introduced a structural framework for Markovian quantum dynamical generators based on a decomposition of the Lindbladian into Hermitian and anti-Hermitian components with respect to the Hilbert--Schmidt inner product. This decomposition naturally separates dissipative and norm-preserving contributions to the dynamics and motivates the introduction of two scalar measures: the dissipative strength $\delta(\mathcal L)$ and the nonnormality $\eta(\mathcal L)$. These quantities admit an operational interpretation in terms of propagator norms, error amplification, and the stability of quantum evolution, thereby providing a unified characterization of Markovian open-system dynamics at the level of the generator.

We established several general structural results. In particular, we showed that normal generators ($\eta(\mathcal L)=0$) exhibit an exact decoupling between dissipative and oscillatory dynamics. This leads to purely exponential behavior of the propagator norm governed by $\delta(\mathcal L)$. Furthermore, we demonstrated that nonnormality is intrinsically linked to dissipation: directional flow in operator space cannot arise in the absence of a finite dissipative component, although the converse does not hold. We further showed that dissipative strength and nonnormality are not independent, but obey structural constraints that can restrict the accessible regions of parameter space. In particular, strong nonnormality requires a separation of scales between coherent and dissipative contributions. These results highlight that nonnormality captures a finer  structural property in operator space  of the generator beyond standard measures of irreversibility.

Building on this framework, we systematically organized quantum dynamical generators into distinct regimes based on the relative scaling of $\delta(\mathcal L)$ and $\eta(\mathcal L)$. The resulting hierarchy, summarized in Table~\ref{tab:regimes}, ranges from purely Hamiltonian dynamics to strongly nonnormal open quantum systems. While Hamiltonian and normal dissipative dynamics exhibit stable, purely exponential behavior, nonnormal systems display transient amplification governed by the dimensionless ratio $\kappa(\mathcal L)=\eta(\mathcal L)/[\delta(\mathcal L)]^2$. In the weakly nonnormal regime, these effects remain perturbative, whereas in the strongly nonnormal regime they dominate the dynamics, leading to enhanced sensitivity to perturbations and increased simulation cost. The crossover regime marks the emergence of this behavior, where dissipative and nonnormal effects contribute comparably.

From the perspective of quantum simulation, this classification provides a transparent connection between generator structure and computational complexity. The dissipative strength $\delta(\mathcal L)$ sets the intrinsic dynamical timescale and baseline cost, while nonnormality introduces additional overhead through transient amplification of errors. As a result, simulation complexity is governed not only by spectral properties but also by the nonnormal geometry of the generator encoded in $\eta(\mathcal L)$.

Several directions for future work naturally emerge from this study. A key open problem is the development of sharper and more general bounds on the amplification factor $A(t)$ in terms of $\eta(\mathcal L)$, particularly beyond perturbative regimes. It would also be valuable to connect the present framework more directly with pseudospectral methods and nonnormal operator theory, potentially leading to refined stability criteria for open quantum dynamics. On the practical side, incorporating these insights into quantum simulation algorithms may enable adaptive schemes that account for transient amplification and optimize resource allocation.

A particularly natural extension concerns explicitly time-dependent
Markovian generators $\mathcal L(t)$. The decomposition into Hermitian and anti-Hermitian components is expected to remain well defined at each instant of time, allowing one to introduce instantaneous measures $\delta(t)\equiv \delta[\mathcal L(t)]$ and $\eta(t)\equiv \eta[\mathcal L(t)]$. However, in contrast to the time-independent setting considered here, the dynamics is then governed by a time-ordered propagator, and a classification based on a single parameter
$\kappa=\eta/\delta^2$ must be replaced by time-dependent or effective quantities that capture the cumulative influence of dissipation and nonnormality over an interval
of evolution. Developing such extensions and determining their relation to transient
amplification represents a particularly promising direction for future research and is currently under investigation.

Another promising direction is the extension of this framework to
time-local generators beyond the strictly Markovian setting,
where memory effects may further enrich the interplay between dissipation
and nonnormality. In such cases the dynamics can still be written in
terms of a time-local generator $\mathcal L(t)$, allowing for an
instantaneous decomposition into dissipative and nonnormal components,
whereas genuinely non-Markovian memory-kernel descriptions require a more
fundamental reformulation of the framework.  Finally, exploring connections
with specific physical platforms—such as driven-dissipative systems,
quantum control protocols, and error-corrected quantum devices—
may provide concrete settings in which the role of nonnormality can be
experimentally probed and exploited.

	\section*{Acknowledgment}
	 We thank the anonymous referees for their comments and suggestions, which helped improve the manuscript.  The author also thanks Laura Cunningham for reading the manuscript.

	\appendix
	\numberwithin{equation}{section}

\section{Anti-Hermitian Superoperators and the Hilbert--Schmidt Inner Product} \label{appn:ah:Lind}

Here we show that for an anti-Hermitian superoperator $\mathcal L_{\rm nd}$ (with respect to the Hilbert--Schmidt inner product), the inner product $\langle X, \mathcal L_{\rm nd}(X) \rangle$ is purely imaginary for any operator $X$. 
Recall that the Hilbert--Schmidt inner product is defined as
\begin{equation} \label{eqn:HS:prod"appn}
	\langle A, B \rangle := \mathrm{Tr}[A^\dagger B], \quad A,B \in \mathcal{B}(\mathcal{H}),
\end{equation}
and a superoperator $\mathcal L_{\rm nd}$ is anti-Hermitian if
\begin{equation}
	\mathcal L_{\rm nd}^\dagger = - \mathcal L_{\rm nd}, \quad \text{i.e. } \langle A, \mathcal L_{\rm nd}(B) \rangle = - \langle \mathcal L_{\rm nd}(A), B \rangle.
\end{equation}
Let $X$ be an arbitrary operator. Then
\begin{equation}
	\langle X, \mathcal L_{\rm nd}(X) \rangle 
	= - \langle \mathcal L_{\rm nd}(X), X \rangle
	= - \langle X, \mathcal L_{\rm nd}(X) \rangle^*.
\end{equation}
Adding the two sides gives
\begin{equation}
	\langle X, \mathcal L_{\rm nd}(X) \rangle + \langle X, \mathcal L_{\rm nd}(X) \rangle^* = 0,
\end{equation}
which implies $\langle X, \mathcal L_{\rm nd}(X) \rangle \in i \mathbb{R}$.
Thus, the inner product of an anti-Hermitian superoperator with any operator is purely imaginary. In particular, the real part vanishes, which explains why the anti-Hermitian component $\mathcal L_{\rm nd}$ does not contribute to the rate of change of the Hilbert--Schmidt norm, \erf{HS:norm:change}.

\section{Operator norm of Hermitian superoperators}
\label{app:hermitian_norm}

We justify why the dissipation strength $\delta(\mathcal L) = \|\mathcal L_{\rm d}\|$ corresponds to the largest eigenvalue magnitude of the Hermitian part $\mathcal L_{\rm d}$ of a Lindbladian.

Let $\mathcal A$ be Hermitian with respect to the Hilbert--Schmidt inner product, \erf{eqn:HS:prod"appn}, so that $\mathcal A^\dagger = \mathcal A$. The induced operator norm is given \erf{eqn:OptNorm:HS}.
By the spectral theorem, $\mathcal A$ admits an orthonormal eigenbasis $\{X_k\}$ with real eigenvalues $\lambda_k$. Expanding $X = \sum_k c_k X_k$, one finds
\begin{equation}
	\frac{\|\mathcal A(X)\|_{\rm HS}^2}{\|X\|_{\rm HS}^2} = \frac{\sum_k |c_k|^2 \lambda_k^2}{\sum_k |c_k|^2},
\end{equation}
which is maximized by choosing $X$ along the eigenvector corresponding to the largest $|\lambda_k|$. Hence
\begin{equation}
	\|\mathcal A\| = \max_k |\lambda_k|.
\end{equation}
Applying this to the Hermitian component of a Lindbladian immediately gives the expression \erf{dis:strength}.

\section{Bounds on propagator norms and error amplification} \label{app:propagation_bounds}
In this appendix we provide supporting derivations for the bounds used in Sec.~\ref{sec:formalism:metrics}.
We begin with the evolution of the Hilbert--Schmidt norm. Recall \erf{HS:norm:change}
\begin{equation} \label{eqn:appnC:EoM}
	\frac{d}{dt} \|\rho(t)\|_{\rm HS}^2
	= 2 \langle \rho, \mathcal L_{\rm d}(\rho) \rangle.
\end{equation}
Using Cauchy–Schwarz inequality and by definition of the Hilbert--Schmidt induced operator norm, \erf{eqn:OptNorm:HS}, one obtains
\begin{equation}
	|\langle \rho, \mathcal L_{\rm d}(\rho) \rangle|
	\le \|\rho\|_{\rm HS} \, \|\mathcal L_{\rm d}(\rho)\|_{\rm HS}
	\le \|\mathcal L_{\rm d}\| \, \|\rho\|_{\rm HS}^2.
\end{equation}
 Substituting this bound into \erf{eqn:appnC:EoM}, and using that $\langle \rho, \mathcal L_{\rm d}(\rho) \rangle$ is real because $\mathcal{L}_{\rm d}$
is Hermitian with respect to the Hilbert-Schmidt inner product, yields 
\begin{equation}
	\frac{d}{dt} \|\rho(t)\|_{\rm HS}^2 \le 2 \|\mathcal L_{\rm d}\| \, \|\rho(t)\|_{\rm HS}^2.
\end{equation}
Applying Gr\"{o}nwall’s inequality yields \erf{eq:prop_bound_main}.

Next, consider the propagator $e^{t\mathcal L}$. For normal generators satisfying $[\mathcal L,\mathcal L^\dagger]=0$, the decomposition $\mathcal L = \mathcal L_{\rm d} + \mathcal L_{\rm nd}$ obeys $[\mathcal L_{\rm d},\mathcal L_{\rm nd}]=0$, and therefore
\begin{equation}
	e^{t\mathcal L} = e^{t\mathcal L_{\rm d}} e^{t\mathcal L_{\rm nd}}.
\end{equation}
Since $\mathcal L_{\rm nd}$ is anti-Hermitian, $e^{t\mathcal L_{\rm nd}}$ is unitary with respect to the Hilbert--Schmidt inner product. Thus, the unitary invariance of the induced norm implies
\begin{equation}
	\left\|e^{t\mathcal L}\right\| = \left\|e^{t\mathcal L_{\rm d}}\right\|.
\end{equation}
Now, applying the spectral theorem to the Hermitian Lindbladian gives
\begin{equation}
	\left\|e^{t\mathcal L_{\rm d}}\right\| = e^{t \lambda_{\max}(\mathcal L_{\rm d})} \le e^{t \|\mathcal L_{\rm d}\|}.
\end{equation}

For nonnormal generators, the behavior of the propagator is more subtle than in the normal case. While the Hermitian component $\mathcal L_{\rm d}$ sets a natural exponential scale for the evolution,
$\delta(\mathcal L) = \|\mathcal L_{\rm d}\|$ defines the maximal instantaneous rate of Hilbert--Schmidt norm growth. However, due to the nonorthogonality of eigenmodes, the propagator may temporarily exceed this rate.
To account for this effect, we introduce a \emph{time-dependent amplification factor} $A(t) \ge 1$ via
\begin{equation} \label{eqn:amp:factor}
	\|e^{t \mathcal L}\| \equiv A(t)\, e^{t\delta(\mathcal L)}.
\end{equation}
Here, $e^{t\delta(\mathcal L)}$ sets the baseline exponential envelope from the dissipative component, while $A(t)$ quantifies transient amplification due to directional flow in operator space generated by $\mathcal L_{\rm nd}$ (in \arf{app:weak_nonnormal_bound} we perturbatively calculate an amplification function). 

This definition recovers the normal case naturally: if $\mathcal L$ is normal ($\eta(\mathcal L) = 0$), all eigenvectors are orthogonal, no transient amplification occurs, and $A(t) = 1$. For nonnormal generators, $A(t) > 1$ reflects temporary constructive interference between nonorthogonal eigenmodes, producing short-term growth in the Hilbert--Schmidt norm beyond what $\delta(\mathcal L)$ alone predicts.
For readers interested in a rigorous mathematical treatment of transient growth in nonnormal semigroups, see the discussion in~\crf{TreEmb05}.

While a general closed-form characterization of $A(t)$ in terms of $\eta(\mathcal L)$ is not available, its behavior can be analyzed in specific regimes. In particular, in the nonnormal case ($\eta(\mathcal L)>0$), transient amplification arises from the noncommutativity between $\mathcal L_{\rm d}$ and $\mathcal L_{\rm nd}$. A systematic analysis of this effect, leading to bounds on $A(t)$ in terms of the dimensionless ratio $\kappa(\mathcal L)=\eta(\mathcal L)/\delta(\mathcal L)^2$, is presented in Sec.~\ref{sec:classes} and Appendix~\ref{app:weak_nonnormal_bound}.

Finally, consider the propagation of numerical errors. Let $\tilde{U}(t)$ be an approximation to the exact propagator $e^{t\mathcal L}$ satisfying
\begin{equation}
	\|\tilde{U}(t) - e^{t\mathcal L}\| \le \epsilon.
\end{equation}
Then, for any initial state $\rho(0)$, submultiplicativity of the induced norm implies
\begin{equation}
	\|\Delta \rho(t) \| 
	\le \|\tilde{U}(t) - e^{t\mathcal L}\| \, \|\rho(0)\|
	\le \epsilon \|\rho(0)\|.
\end{equation}
where $\|\Delta \rho(t) \| \coloneq \|\tilde{U}(t)\rho(0) - \rho(t)\|$.
In practical simulation schemes, however, errors accumulate over intermediate time steps (e.g., in product formulas or time discretization). These errors are subsequently propagated and potentially amplified by the dynamics itself. As a result, the total error is controlled by the norm of the propagator, yielding
\begin{equation}
	\|\Delta \rho(t)\| \lesssim \epsilon \,\left\|e^{t\mathcal L}\right\|
	\le \epsilon \, A(t)\, e^{t \|\mathcal L_{\rm d}\|}.
\end{equation}
where we used \erf{eqn:amp:factor}. The first inequality indicates that the accumulated error is controlled by the propagator norm times $\epsilon$, up to constants depending on the accumulation procedure.

This shows that numerical errors are amplified not only by the dissipative scale set by $\delta(\mathcal L)$, but also by the transient amplification factor $A(t)$ arising from nonnormality. In particular, for normal generators ($\eta(\mathcal L)=0$), one has $A(t)=1$, and no additional amplification occurs. By contrast, for nonnormal generators, $A(t)>1$ can increase the effective precision requirements of simulation algorithms.

\section{Multi-qubit dephasing: structure and norm scaling} \label{app:multi_qubit_dephasing}

Consider an $ K$-qubit system subject to independent pure dephasing on each qubit, described by the Lindbladian (for brevity, we denote $\mathcal L_{\rm z}^K(\rho)$ by $\mathcal L(\rho)$ in the following):
\begin{equation}
	\mathcal L(\rho) = \sum_{k=1}^N \gamma_k \left( \hat\sigma_{\rm z}^{(k)} \rho \hat\sigma_{\rm z}^{(k)} - \rho \right),
\end{equation}
where $\hat\sigma_{\rm z}^{(k)}$ denotes the Pauli-$z$ operator acting on qubit $k$ and $\gamma_k \ge 0$ are the corresponding dephasing rates. Each local contribution
\begin{equation}
	\mathcal L_k(\rho) = \gamma_k \left( \hat\sigma_{\rm z}^{(k)} \rho \hat\sigma_{\rm z}^{(k)} - \rho \right)
\end{equation}
is Hermitian with respect to the Hilbert--Schmidt inner product, since $\hat\sigma_{\rm z}^{(k)\dagger} = \hat\sigma_{\rm z}^{(k)}$ and conjugation by $\hat\sigma_{\rm z}^{(k)}$ is unitary. It follows that $\mathcal L_k^\dagger = \mathcal L_k$, and therefore the full generator satisfies $\mathcal L^\dagger = \mathcal L$, implying $\mathcal L = \mathcal L_{\rm d}$ and $\mathcal L_{\rm nd}=0$. As a result, $\eta(\mathcal L) = [\mathcal L, \mathcal L^\dagger] = 0$,  showing that independent dephasing defines a normal Lindbladian for any system size.

To evaluate the dissipative strength, we compute the induced operator norm, \erf{eqn:OptNorm:HS}. A convenient basis for operator space is given by tensor products of Pauli operators $\{\hat I,\hat\sigma_{\rm x},\hat\sigma_{\rm y},\hat\sigma_{\rm z}\}^{\otimes  K}$, which form an orthogonal basis under the Hilbert--Schmidt inner product. Each basis element $\hat P = \bigotimes_{k=1}^{ K} \hat P_k$ with $\hat P_k \in \{\hat I,\hat\sigma_{\rm x},\hat\sigma_{\rm y},\hat\sigma_{\rm z}\}$ is an eigenoperator of $\mathcal L$. Indeed, using $\hat\sigma_{\rm z}^{(k)} \hat P \hat\sigma_{\rm z}^{(k)} = \pm \hat P$, with a minus sign whenever $\hat P_k \in \{\hat\sigma_{\rm x},\hat\sigma_{\rm y}\}$ and a plus sign when $\hat P_k \in \{\hat I,\hat\sigma_{\rm z}\}$, one finds
\begin{equation}
	\mathcal L_k(\hat P) =
	\begin{cases}
		0, & \hat P_k \in \{\hat I,\hat\sigma_{\rm z}\}, \\
		-2\gamma_k \hat P, & \hat P_k \in \{\hat\sigma_{\rm x},\hat\sigma_{\rm y}\}.
	\end{cases}
\end{equation}
Summing over all qubits yields
\begin{equation}
	\mathcal L(\hat P) = -2 \sum_{k \in S(\hat P)} \gamma_k \, \hat P,
\end{equation}
where $S(\hat P)$ denotes the set of qubits on which $\hat P_k \in \{\hat\sigma_{\rm x},\hat\sigma_{\rm y}\}$. Thus, each Pauli string is an eigenoperator with eigenvalue $-2 \sum_{k \in S(\hat P)} \gamma_k$. The operator norm is therefore given by the maximal magnitude of these eigenvalues, which is attained when $\hat P_k \in \{\hat\sigma_{\rm x},\hat\sigma_{\rm y}\}$ for all qubits, yielding
\begin{equation}
	\delta(\mathcal L) = \|\mathcal L\| = 2 \sum_{k=1}^{ K} \gamma_k.
\end{equation}
Hence, the dissipative strength scales additively with system size,
while the nonnormality remains identically zero. 

 The additive scaling found here follows from the locality and diagonal action of the dephasing generator in the Pauli basis. Similar behavior is expected for other collections of commuting local dissipators, although a general classification of norm scaling for arbitrary normal Lindbladians is beyond the scope of this work. 

\section{Dissipative strength for Pauli channels} \label{appn:1qb:1pauli}
We consider a single-qubit Pauli channel of the form
\begin{equation}
	\mathcal L(\rho)
	= \sum_{\alpha \in \{{\rm x,y,z}\}} \gamma_\alpha \left(
	\hat \sigma_{\alpha} \rho \hat \sigma_{\alpha} - \rho
	\right),
\end{equation}
with rates $\gamma_\alpha \ge 0$.

The Pauli operators $\{\hat I, \hat \sigma_{\rm x}, \hat \sigma_{\rm y}, \hat \sigma_{\rm z}\}$ form an orthogonal basis with respect to the Hilbert--Schmidt inner product, and the generator acts diagonally in this basis. In particular,
\begin{equation}
	\mathcal L(\hat \sigma_\beta)
	= \lambda_\beta \hat \sigma_\beta.
\end{equation}

Using the identity
\begin{equation}
	\hat \sigma_\alpha \hat \sigma_\beta \hat \sigma_\alpha =
	\begin{cases}
		\hat \sigma_\beta, & \alpha = \beta, \\
		-\hat \sigma_\beta, & \alpha \neq \beta,
	\end{cases}
\end{equation}
one finds the eigenvalues
\begin{align}
	\lambda_{\rm x} &= -2(\gamma_{\rm y} + \gamma_{\rm z}), \\
	\lambda_{\rm y} &= -2(\gamma_{\rm x} + \gamma_{\rm z}), \\
	\lambda_{\rm z} &= -2(\gamma_{\rm x} + \gamma_{\rm y}),
\end{align}
while $\lambda_0 = 0$ for the identity operator.

Since the generator is normal, the dissipative strength is given by
\begin{equation}
	\delta(\mathcal L)
	= \max_\alpha |\mathrm{Re}(\lambda_\alpha)|,
\end{equation}
which yields
\begin{equation}
	\delta(\mathcal L)
	= 2 \max \left\{
	\gamma_{\rm y} + \gamma_{\rm z},\;
	\gamma_{\rm x} + \gamma_{\rm z},\;
	\gamma_{\rm x} + \gamma_{\rm y}
	\right\}.
\end{equation}

\section{Scaling regimes of nonnormal dynamics}
\label{app:nonnormal_regimes}

In this appendix, we briefly justify the  scaling  regimes introduced in Sec.~\ref{sec:nonnormal} for generators with $\delta(\mathcal L)>0$ and $\eta(\mathcal L)>0$.

The distinction between weakly and strongly nonnormal dynamics can be understood by examining the interplay between the Hermitian and anti-Hermitian components of the generator, $\mathcal L = \mathcal L_{\rm d} + \mathcal L_{\rm nd}$. The dissipative strength $\delta(\mathcal L)$ sets the characteristic decay scale, while the nonnormality $\eta(\mathcal L)$ quantifies the extent to which these components fail to commute.
At short times, the propagator admits the expansion
\begin{equation}
	e^{t\mathcal L} = \mathcal I + t\mathcal L + \tfrac{t^2}{2}\mathcal L^2 + O(t^3).
\end{equation}
Expanding the quadratic term using $\mathcal L = \mathcal L_{\rm d} + \mathcal L_{\rm nd}$ yields
\begin{equation}
	\mathcal L^2
	=
	\mathcal L_{\rm d}^2
	+ \mathcal L_{\rm nd}^2
	+ \{\mathcal L_{\rm d},\mathcal L_{\rm nd}\}
	+ [\mathcal L_{\rm nd},\mathcal L_{\rm d}].
\end{equation}
The first three terms contribute to the baseline dissipative and norm-preserving evolution. In contrast, the commutator encodes the leading-order nonnormal contribution. Its magnitude is controlled by $\eta(\mathcal L)$ as a norm measure of noncommutativity.

To assess the relative importance of these contributions, we consider timescales $t \sim 1/\delta(\mathcal L)$ determined by dissipative decay. At such times, the leading dissipative term contributes at order unity, while the commutator-induced correction scales as
\begin{equation}
	t^2 \,\|[\mathcal L_{\rm d},\mathcal L_{\rm nd}]\|
	\;\sim\;
	\frac{\eta(\mathcal L)}{\delta(\mathcal L)^2}.
\end{equation}
This identifies the dimensionless ratio $\kappa(\mathcal L)$
as the parameter governing the relative importance of nonnormal effects.

When $\kappa(\mathcal L) \ll 1$, the commutator contribution remains perturbative over the dissipative timescale, and the dynamics is well approximated by that of a normal generator. In contrast, when $\kappa(\mathcal L) \gtrsim 1$, nonnormal effects become comparable to or dominate the leading dissipative behavior, enabling transient amplification and enhanced sensitivity to perturbations.

This provides the basis for the regime classification used in Sec.~\ref{sec:classes}. A more detailed quantitative characterization can be obtained using bounds on the propagator norm and pseudospectral analysis, which we do not pursue here.

\section{Perturbative bound for nonnormal generators}
\label{app:weak_nonnormal_bound}

In this appendix we justify the scaling behavior of the propagator norm for generators of the form given  in \erf{L:d:nd}.
 We assume that the commutator $[\mathcal L_{\rm d},\mathcal L_{\rm nd}]$ is nonvanishing, so that the generator is nonnormal. Let us define the propagator $\hat U(t) = e^{t\mathcal L}$ and introduce the interaction-picture operator
\begin{equation}
	\hat V(t) \coloneq e^{-t\mathcal L_{\rm d}} e^{t\mathcal L},
\end{equation}
such that $\hat U(t) = e^{t\mathcal L_{\rm d}} \hat V(t)$. A direct differentiation shows that $\hat V(t)$ satisfies the differential equation
\begin{equation}
	\frac{d}{dt}\hat V(t)
	=
	\widetilde{\mathcal L}_{\rm nd}(t)\,\hat V(t),
\end{equation}
where $	\widetilde{\mathcal L}_{\rm nd}(t) \coloneq e^{-t\mathcal L_{\rm d}} \mathcal L_{\rm nd} e^{t\mathcal L_{\rm d}}$. The solution is given by the time-ordered exponential
\begin{equation}
	\hat V(t)
	=
	\mathfrak T \left[\exp\!\left( \int_0^t ds \, \widetilde{\mathcal L}_{\rm nd}(s) \right)\right],
\end{equation}
and therefore
\begin{equation}
	e^{t\mathcal L}
	=
	e^{t\mathcal L_{\rm d}}
	\,
	\mathfrak T \left[\exp\!\left( \int_0^t ds \, \widetilde{\mathcal L}_{\rm nd}(s) \right)\right].
\end{equation}
The superoperator $\widetilde{\mathcal L}_{\rm nd}(s)$ can be expressed using the identity
\begin{equation}
	e^{-s\hat P} \hat Q e^{s \hat P}
	=
	\sum_{n=0}^{\infty} \frac{(- s)^{ n}}{n!} \, \mathrm{ad}_{\hat P}^n(\hat Q),
	\qquad
	\mathrm{ad}_{\hat P}(\hat Q) \coloneq [\hat P,\hat Q],
\end{equation}
 where $\mathrm{ad}_{\hat P}^n(\hat Q)$ denotes the $n$-fold iterated adjoint action, defined recursively by $\mathrm{ad}_{\hat P}^0(\hat Q) = \hat Q$, and
\begin{equation}
	\mathrm{ad}_{\hat P}^n(\hat Q) = [\hat P,\mathrm{ad}_{\hat P}^{n-1}(\hat Q)],
\end{equation}

which yields
\begin{equation}
	\widetilde{\mathcal L}_{\rm nd}(s)
	=
	\mathcal L_{\rm nd}
	+ s [\mathcal L_{\rm nd}, \mathcal L_{\rm d}]
	+ \frac{s^2}{2} [[\mathcal L_{\rm nd},\mathcal L_{\rm d}],\mathcal L_{\rm d}]
	+ \cdots.
\end{equation}
This expansion makes explicit that deviations from time-independence are governed by nested commutators involving $\mathcal L_{\rm d}$ and $\mathcal L_{\rm nd}$.

We now bound the operator norm of the propagator. Using submultiplicativity and the interaction-picture decomposition,
\begin{equation}
	\|e^{t\mathcal L}\|
	\le
	\|e^{t\mathcal L_{\rm d}}\|
	\,
	\left\|
	\mathfrak T \left[\exp\!\left( \int_0^t ds \, \widetilde{\mathcal L}_{\rm nd}(s) \right)\right]
	\right\|.
\end{equation}
Since $\mathcal L_{\rm d}$ is Hermitian,  it is unitarily diagonalizable with real eigenvalues. Therefore, by the spectral theorem, its induced operator norm is determined by its largest eigenvalue, yielding 
\begin{equation}
	\|e^{t\mathcal L_{\rm d}}\|
	=
	e^{t \|\mathcal L_{\rm d}\|}
	=
	e^{t \delta(\mathcal L)}.
\end{equation}
For the time-ordered exponential we use the general inequality
\begin{equation} \label{gen:ineq:time:order}
	\left\|
	\mathfrak T \left[\exp\!\left( \int_0^t ds\, \hat B(s) \right)\right]
	\right\|
	\le
	\exp\!\left( \int_0^t ds\, \|\hat B(s)\| \right),
\end{equation}
valid for bounded operators $\hat B(s)$.
Using the nested commutator expansion and the triangle inequality,
\begin{equation}
	\|\widetilde{\mathcal L}_{\rm nd}(s)\|
	\le
	\|\mathcal L_{\rm nd}\|
	+ s \|[\mathcal L_{\rm nd}, \mathcal L_{\rm d}]\|
	+ \frac{s^2}{2} \|[[\mathcal L_{\rm nd},\mathcal L_{\rm d}],\mathcal L_{\rm d}]\|
	+ \cdots.
\end{equation}
Now, using the nonnormality definition, \erf{nn:gen} and \erf{L:d:nd}, we get $\eta(\mathcal L) = 2\|[\mathcal L_{\rm d}, \mathcal L_{\rm nd}]\|$,
and assuming that higher nested commutators do not exceed the natural scale set by $\eta(\mathcal L)$ and $\|\mathcal L_{\rm d}\|$, we obtain the estimate
\begin{equation}
	\|\widetilde{\mathcal L}_{\rm nd}(s)\|
	\le
	\|\mathcal L_{\rm nd}\|
	+ s/2\, \eta(\mathcal L)
	+ O(s^2 \|\mathcal L_{\rm d}\| \eta(\mathcal L)).
\end{equation}
Integrating over $s \in [0,t]$ and substituting into \erf{gen:ineq:time:order} yields
\begin{equation}
	\left\|
	\mathfrak T e^{\left( \int_0^t ds \, \widetilde{\mathcal L}_{\rm nd}(s) \right)}
	\right\|
	\le
	\exp\!\big(
	t \|\mathcal L_{\rm nd}\|
	+ \tfrac{1}{4} t^2 \eta(\mathcal L)
	+ \cdots
	\big),
\end{equation}
where the ellipsis denotes higher-order contributions arising from nested commutators.
While these terms are higher order in $t$ for short times, they do not form a
parametrically suppressed hierarchy at timescales $t \sim 1/\delta(\mathcal L)$.
Instead, all such contributions scale as $\eta/\delta^2$,
so that the deviation from purely exponential behavior is controlled by this
dimensionless ratio. 
Also, since $\mathcal L_{\rm nd}$ is anti-Hermitian, it generates a norm-preserving flow in the Hilbert--Schmidt inner product, and therefore does not contribute to exponential growth of the propagator norm. Consequently, the linear term $t\|\mathcal L_{\rm nd}\|$ is not associated with any physical amplification mechanism; rather, it arises from a crude application of the triangle inequality in bounding $\|\widetilde{\mathcal L}_{\rm nd}(s)\|$. The actual growth of the interaction-picture dynamics is therefore governed by the commutator structure encoded in $[\mathcal L_{\rm d},\mathcal L_{\rm nd}]$, while the anti-Hermitian part only contributes oscillatory (non-amplifying) rotations in operator space.
Thus, the deviation from purely exponential behavior is controlled by this dimensionless ratio, which serves as the natural
perturbative parameter in the weakly nonnormal regime. Then, 
\begin{equation}
	\|e^{t\mathcal L}\|
	\le
	e^{t \delta}
	\exp\left[ O(\eta/\delta^2) \right].
\end{equation}
Therefore, when $\eta \ll \delta^2$, the effect of nonnormality remains perturbative in the dimensionless ratio, and the propagator norm follows the exponential scaling governed by $\delta(\mathcal L)$ up to controlled multiplicative corrections. When $\eta \gtrsim \delta^2$, however, the factor $\exp\!\left( O[\kappa(\mathcal L)] \right)$ modulates the exponential profile, reflecting the growing influence of nonnormal effects. 
This factor provides an estimate for the amplification amplitude $A(t)$ introduced in~\arf{app:propagation_bounds}.

\section{Competing dephasing and relaxation: derivation of $\delta(\mathcal L)$ and $\eta(\mathcal L)$}
\label{app:dephasing_relaxation}

In this appendix we derive the scaling of the dissipative strength $\delta(\mathcal L)$ and the nonnormality $\eta(\mathcal L)$ for a single qubit subject to both dephasing and relaxation, \erf{rates:dephase:relax}. Here for brevity we drop the subindex $\mathcal L_{\rm z-}$ in the Lindbladian.
Let us consider a Lindbladian of the form $\mathcal L = \mathcal L_{\rm z} + \mathcal L_{-}$,
where $\mathcal L_{\rm z}(\rho)= \gamma_{\rm z} \left( \hat\sigma_{\rm z} \rho \hat\sigma_{\rm z} - \rho \right)$ and $\mathcal L_{-}(\rho) = \gamma_- \left( \hat\sigma_- \rho \hat\sigma_+ - \tfrac{1}{2}\{\hat\sigma_+\hat\sigma_-, \rho\} \right)$.
\begin{equation}
	\mathcal L(\hat\sigma_{\rm x}) = -(2\gamma_{\rm z} + \tfrac{\gamma_-}{2}) \hat\sigma_{\rm x},
	\quad
	\mathcal L(\hat\sigma_{\rm y}) = -(2\gamma_{\rm z} + \tfrac{\gamma_-}{2}) \hat\sigma_{\rm y},
\end{equation}
while $\hat\sigma_{\rm z}$ mixes with the identity due to relaxation.
The dephasing part $\mathcal L_{\rm z}$ is Hermitian, while $\mathcal L_{-}$ contributes both Hermitian and anti-Hermitian components. It is easy to show that the Hermitian part of $\mathcal L_{-}$ acts diagonally on $\hat\sigma_{\rm x},\hat\sigma_{\rm y}$ with rate $\gamma_-/2$, and contributes an $O(\gamma_-)$ term on $\hat\sigma_{\rm z}$.
Therefore, the largest decay rate in magnitude is
\begin{equation}
	\delta(\mathcal L) = \max \left\{ 2\gamma_{\rm z} + \tfrac{\gamma_-}{2}, \; \gamma_- \right\}.
\end{equation}
  A direct computation in the Pauli basis $\{\hat I, \hat \sigma_{\rm x},\hat \sigma_{\rm y},\hat \sigma_{\rm z}\}$ shows that the full Lindbladian can be written as
\begin{equation}
	\mathcal L =
	\begin{pmatrix}
		0 & 0 & 0 & 0 \\
		0 & -(2\gamma_{\rm z} + \gamma_-/2) & 0 & 0 \\
		0 & 0 & -(2\gamma_{\rm z} + \gamma_-/2) & 0 \\
		-\gamma_- & 0 & 0 & -\gamma_-
	\end{pmatrix}.
\end{equation}
Using this representation, one obtains
\begin{equation}
	[\mathcal L,\mathcal L^\dagger]
	=
	\gamma_-^2
	\begin{pmatrix}
		-1 & 0 & 0 & -1 \\
		0 & 0 & 0 & 0 \\
		0 & 0 & 0 & 0 \\
		-1 & 0 & 0 & 1
	\end{pmatrix},
\end{equation}
whose nonzero eigenvalues are $\pm \sqrt{2}\,\gamma_-^2$. Therefore, the nonnormality measure is
\begin{equation}
	\eta(\mathcal L)
	=
	\|[\mathcal L,\mathcal L^\dagger]\|
	=
	\sqrt{2}\,\gamma_-^2.
\end{equation}

Importantly, $\eta(\mathcal L)$ is independent of the dephasing rate $\gamma_{\rm z}$; the latter contributes only to the dissipative strength.

\section{Driven amplitude damping: emergence of strong nonnormality}
\label{app:driven_amplitude_damping}

We coensider a single qubit subject to amplitude damping at rate $\gamma_-$ and a coherent resonant drive of strength $\Omega$. The Lindbladian in \erf{eqn:GKSL} is constructed by setting $\hat L_- = \sqrt{\gamma_-}\,\hat \sigma_-$, and $\hat H = \frac{\Omega}{2} \hat \sigma_{\rm x}$. In the Pauli basis, the dissipative and anti-Hermitian parts are expressed as 
\begin{equation}
	\mathcal L_{\rm d} \equiv
	\begin{pmatrix}
		0 & 0 & 0 & -\gamma_-/2 \\
		0 & -\gamma_-/2 & 0 & 0 \\
		0 & 0 & -\gamma_-/2 & 0 \\
		-\gamma_-/2 & 0 & 0 & -\gamma_-
	\end{pmatrix},
\end{equation}
and 
\begin{equation}
	\mathcal L_{\rm nd} \equiv
	\begin{pmatrix}
		0 & 0 & 0 & \gamma_-/2 \\
		0 & 0 & 0 & 0 \\
		0 & 0 & 0 & -\Omega \\
		-\gamma_-/2 & 0 & \Omega & 0
	\end{pmatrix},
\end{equation}
respectively. The dissipative strength is determined by the largest eigenvalue magnitude of
$\mathcal L_{\rm d}$. The spectrum of $\mathcal L_{\rm d}$ is $\frac{-\gamma_-}{2}\left\{1,1,\left(1\pm\sqrt{2}\right) \right\}$, and therefore
\begin{equation}
	\delta(\mathcal L) =
	\frac{1+\sqrt{2}}{2}\,\gamma_- \sim \gamma_-.
\end{equation}
The nonnormality measure is defined as $\eta(\mathcal L) 
= 2\|[\mathcal L_{\rm d},\mathcal L_{\rm nd}]\|$. Using the above matrices 
\begin{equation}
	[\mathcal L_{\rm d},\mathcal L_{\rm nd}]
	=
	\begin{pmatrix}
		{\gamma_-^2}/{2} & 0 & -{\Omega\gamma_-}/{2} & {\gamma_-^2}/{2}\\
		0&0&0&0\\
		-{\Omega\gamma_-}/{2}&0&0&-{\Omega\gamma_-}/{2}\\
		{\gamma_-^2}/{2}&0&{\Omega\gamma_-}/{2}&{\gamma_-^2}/{2}
	\end{pmatrix}. 
\end{equation}

The $\gamma_-^2$ terms reflect the intrinsic nonnormality of amplitude damping, which persists even in the absence of coherent driving. The
off-diagonal terms proportional to $\gamma_-\Omega$ originate from the noncommutativity between coherent rotations generated by $\mathcal L_{\rm nd}$ and dissipative relaxation generated by $\mathcal L_{\rm d}$, producing a shear in operator space. Consequently, in the strong-driving limit $\Omega\gg\gamma_-$, the operator norm of the commutator is dominated by the $\gamma_-\Omega$ terms, yielding $\|[\mathcal L_{\rm d},\mathcal L_{\rm nd}]\| \sim \gamma_- \Omega$, and hence
\begin{equation}
	\eta(\mathcal L) \sim \gamma_-^2 + \gamma_- \Omega.
\end{equation}
Therefore, the strong nonnormal regime $\kappa(\mathcal L) = \eta(\mathcal L)/[\delta(\mathcal L)]^2\gg 1$ is achieved when $\Omega \gg \gamma_-$, i.e., in the strong-driving limit.

\bibliography{references}

\providecommand{\noopsort}[1]{}\providecommand{\singleletter}[1]{#1}%
\begin{thebibliography}{41}%
\makeatletter
\providecommand \@ifxundefined [1]{%
 \@ifx{#1\undefined}
}%
\providecommand \@ifnum [1]{%
 \ifnum #1\expandafter \@firstoftwo
 \else \expandafter \@secondoftwo
 \fi
}%
\providecommand \@ifx [1]{%
 \ifx #1\expandafter \@firstoftwo
 \else \expandafter \@secondoftwo
 \fi
}%
\providecommand \natexlab [1]{#1}%
\providecommand \enquote  [1]{``#1''}%
\providecommand \bibnamefont  [1]{#1}%
\providecommand \bibfnamefont [1]{#1}%
\providecommand \citenamefont [1]{#1}%
\providecommand \href@noop [0]{\@secondoftwo}%
\providecommand \href [0]{\begingroup \@sanitize@url \@href}%
\providecommand \@href[1]{\@@startlink{#1}\@@href}%
\providecommand \@@href[1]{\endgroup#1\@@endlink}%
\providecommand \@sanitize@url [0]{\catcode `\\12\catcode `\$12\catcode
  `\&12\catcode `\#12\catcode `\^12\catcode `\_12\catcode `\%12\relax}%
\providecommand \@@startlink[1]{}%
\providecommand \@@endlink[0]{}%
\providecommand \url  [0]{\begingroup\@sanitize@url \@url }%
\providecommand \@url [1]{\endgroup\@href {#1}{\urlprefix }}%
\providecommand \urlprefix  [0]{URL }%
\providecommand \Eprint [0]{\href }%
\providecommand \doibase [0]{https://doi.org/}%
\providecommand \selectlanguage [0]{\@gobble}%
\providecommand \bibinfo  [0]{\@secondoftwo}%
\providecommand \bibfield  [0]{\@secondoftwo}%
\providecommand \translation [1]{[#1]}%
\providecommand \BibitemOpen [0]{}%
\providecommand \bibitemStop [0]{}%
\providecommand \bibitemNoStop [0]{.\EOS\space}%
\providecommand \EOS [0]{\spacefactor3000\relax}%
\providecommand \BibitemShut  [1]{\csname bibitem#1\endcsname}%
\let\auto@bib@innerbib\@empty
\bibitem [{\citenamefont {Breuer}\ and\ \citenamefont
  {Petruccione}(2007)}]{BrePet02}%
  \BibitemOpen
  \bibfield  {author} {\bibinfo {author} {\bibfnamefont {H.~P.}\ \bibnamefont
  {Breuer}}\ and\ \bibinfo {author} {\bibfnamefont {F.}~\bibnamefont
  {Petruccione}},\ }\href
  {https://doi.org/10.1093/acprof:oso/9780199213900.001.0001} {\emph {\bibinfo
  {title} {The Theory of Open Quantum Systems}}}\ (\bibinfo  {publisher}
  {Oxford University Press},\ \bibinfo {year} {2007})\BibitemShut {NoStop}%
\bibitem [{\citenamefont {Rivas}\ and\ \citenamefont
  {Huelga}(2012)}]{RivHue12}%
  \BibitemOpen
  \bibfield  {author} {\bibinfo {author} {\bibfnamefont {A.}~\bibnamefont
  {Rivas}}\ and\ \bibinfo {author} {\bibfnamefont {S.~F.}\ \bibnamefont
  {Huelga}},\ }\href {https://doi.org/10.1007/978-3-642-23354-8} {\emph
  {\bibinfo {title} {Open Quantum Systems: An Introduction}}}\ (\bibinfo
  {publisher} {Springer},\ \bibinfo {year} {2012})\BibitemShut {NoStop}%
\bibitem [{\citenamefont {Gardiner}\ and\ \citenamefont
  {Zoller}(2004)}]{GarZol04}%
  \BibitemOpen
  \bibfield  {author} {\bibinfo {author} {\bibfnamefont {C.~W.}\ \bibnamefont
  {Gardiner}}\ and\ \bibinfo {author} {\bibfnamefont {P.}~\bibnamefont
  {Zoller}},\ }\href {https://link.springer.com/book/9783540223016} {\emph
  {\bibinfo {title} {Quantum Noise}}}\ (\bibinfo  {publisher} {Springer},\
  \bibinfo {year} {2004})\BibitemShut {NoStop}%
\bibitem [{\citenamefont {Weimer}(2015)}]{Weimer15}%
  \BibitemOpen
  \bibfield  {author} {\bibinfo {author} {\bibfnamefont {H.}~\bibnamefont
  {Weimer}},\ }\bibfield  {title} {\bibinfo {title} {Variational principle for
  steady states of dissipative quantum many-body systems},\ }\href
  {https://doi.org/10.1103/PhysRevLett.114.040402} {\bibfield  {journal}
  {\bibinfo  {journal} {Phys. Rev. Lett.}\ }\textbf {\bibinfo {volume} {114}},\
  \bibinfo {pages} {040402} (\bibinfo {year} {2015})}\BibitemShut {NoStop}%
\bibitem [{\citenamefont {Schlimgen}\ \emph {et~al.}(2021)\citenamefont
  {Schlimgen}, \citenamefont {Head-Marsden}, \citenamefont {Sager},
  \citenamefont {Narang},\ and\ \citenamefont {Mazziotti}}]{SchMaz21}%
  \BibitemOpen
  \bibfield  {author} {\bibinfo {author} {\bibfnamefont {A.~W.}\ \bibnamefont
  {Schlimgen}}, \bibinfo {author} {\bibfnamefont {K.}~\bibnamefont
  {Head-Marsden}}, \bibinfo {author} {\bibfnamefont {L.~M.}\ \bibnamefont
  {Sager}}, \bibinfo {author} {\bibfnamefont {P.}~\bibnamefont {Narang}},\ and\
  \bibinfo {author} {\bibfnamefont {D.~A.}\ \bibnamefont {Mazziotti}},\
  }\bibfield  {title} {\bibinfo {title} {Quantum simulation of open quantum
  systems using a unitary decomposition of operators},\ }\href
  {https://doi.org/10.1103/PhysRevLett.127.270503} {\bibfield  {journal}
  {\bibinfo  {journal} {Phys. Rev. Lett.}\ }\textbf {\bibinfo {volume} {127}},\
  \bibinfo {pages} {270503} (\bibinfo {year} {2021})}\BibitemShut {NoStop}%
\bibitem [{\citenamefont {Kamakari}\ \emph {et~al.}(2022)\citenamefont
  {Kamakari}, \citenamefont {Sun}, \citenamefont {Motta},\ and\ \citenamefont
  {Minnich}}]{KamMin22}%
  \BibitemOpen
  \bibfield  {author} {\bibinfo {author} {\bibfnamefont {H.}~\bibnamefont
  {Kamakari}}, \bibinfo {author} {\bibfnamefont {S.-N.}\ \bibnamefont {Sun}},
  \bibinfo {author} {\bibfnamefont {M.}~\bibnamefont {Motta}},\ and\ \bibinfo
  {author} {\bibfnamefont {A.~J.}\ \bibnamefont {Minnich}},\ }\bibfield
  {title} {\bibinfo {title} {Digital quantum simulation of open quantum systems
  using quantum imaginary--time evolution},\ }\href
  {https://doi.org/10.1103/PRXQuantum.3.010320} {\bibfield  {journal} {\bibinfo
   {journal} {PRX Quantum}\ }\textbf {\bibinfo {volume} {3}},\ \bibinfo {pages}
  {010320} (\bibinfo {year} {2022})}\BibitemShut {NoStop}%
\bibitem [{\citenamefont {Suri}\ \emph {et~al.}(2023)\citenamefont {Suri},
  \citenamefont {Barreto}, \citenamefont {Hadfield}, \citenamefont {Wiebe},
  \citenamefont {Wudarski},\ and\ \citenamefont {Marshall}}]{SurMar23}%
  \BibitemOpen
  \bibfield  {author} {\bibinfo {author} {\bibfnamefont {N.}~\bibnamefont
  {Suri}}, \bibinfo {author} {\bibfnamefont {J.}~\bibnamefont {Barreto}},
  \bibinfo {author} {\bibfnamefont {S.}~\bibnamefont {Hadfield}}, \bibinfo
  {author} {\bibfnamefont {N.}~\bibnamefont {Wiebe}}, \bibinfo {author}
  {\bibfnamefont {F.}~\bibnamefont {Wudarski}},\ and\ \bibinfo {author}
  {\bibfnamefont {J.}~\bibnamefont {Marshall}},\ }\bibfield  {title} {\bibinfo
  {title} {Two-unitary decomposition algorithm and open quantum system
  simulation},\ }\href {https://doi.org/10.22331/q-2023-05-15-1002} {\bibfield
  {journal} {\bibinfo  {journal} {{Quantum}}\ }\textbf {\bibinfo {volume}
  {7}},\ \bibinfo {pages} {1002} (\bibinfo {year} {2023})}\BibitemShut
  {NoStop}%
\bibitem [{\citenamefont {Pocrnic}\ \emph {et~al.}(2025)\citenamefont
  {Pocrnic}, \citenamefont {Segal},\ and\ \citenamefont {Wiebe}}]{PocWie25}%
  \BibitemOpen
  \bibfield  {author} {\bibinfo {author} {\bibfnamefont {M.}~\bibnamefont
  {Pocrnic}}, \bibinfo {author} {\bibfnamefont {D.}~\bibnamefont {Segal}},\
  and\ \bibinfo {author} {\bibfnamefont {N.}~\bibnamefont {Wiebe}},\ }\bibfield
   {title} {\bibinfo {title} {Quantum simulation of {L}indbladian dynamics via
  repeated interactions},\ }\href {https://arxiv.org/abs/2312.05371} {\bibfield
   {journal} {\bibinfo  {journal} {arXiv:2312.05371}\ } (\bibinfo {year}
  {2025})}\BibitemShut {NoStop}%
\bibitem [{\citenamefont {Sander}\ \emph {et~al.}(2025)\citenamefont {Sander},
  \citenamefont {Fr{\"o}hlich}, \citenamefont {Eigel}, \citenamefont {Eisert},
  \citenamefont {Gel{\ss}}, \citenamefont {Hinterm{\"u}ller}, \citenamefont
  {Milbradt}, \citenamefont {Wille},\ and\ \citenamefont {Mendl}}]{SanMen25}%
  \BibitemOpen
  \bibfield  {author} {\bibinfo {author} {\bibfnamefont {A.}~\bibnamefont
  {Sander}}, \bibinfo {author} {\bibfnamefont {M.}~\bibnamefont
  {Fr{\"o}hlich}}, \bibinfo {author} {\bibfnamefont {M.}~\bibnamefont {Eigel}},
  \bibinfo {author} {\bibfnamefont {J.}~\bibnamefont {Eisert}}, \bibinfo
  {author} {\bibfnamefont {P.}~\bibnamefont {Gel{\ss}}}, \bibinfo {author}
  {\bibfnamefont {M.}~\bibnamefont {Hinterm{\"u}ller}}, \bibinfo {author}
  {\bibfnamefont {R.~M.}\ \bibnamefont {Milbradt}}, \bibinfo {author}
  {\bibfnamefont {R.}~\bibnamefont {Wille}},\ and\ \bibinfo {author}
  {\bibfnamefont {C.~B.}\ \bibnamefont {Mendl}},\ }\bibfield  {title} {\bibinfo
  {title} {Large-scale stochastic simulation of open quantum systems},\ }\href
  {https://doi.org/10.1038/s41467-025-66846-x} {\bibfield  {journal} {\bibinfo
  {journal} {Nature Communications}\ }\textbf {\bibinfo {volume} {16}},\
  \bibinfo {pages} {11074} (\bibinfo {year} {2025})}\BibitemShut {NoStop}%
\bibitem [{\citenamefont {Liu}\ \emph {et~al.}(2025)\citenamefont {Liu},
  \citenamefont {Lin}, \citenamefont {Chen}, \citenamefont {Xue}, \citenamefont
  {Sun}, \citenamefont {Li}, \citenamefont {Zhuang}, \citenamefont {Wang},
  \citenamefont {Wu}, \citenamefont {Gong},\ and\ \citenamefont
  {Guo}}]{LiuGuo25}%
  \BibitemOpen
  \bibfield  {author} {\bibinfo {author} {\bibfnamefont {H.~Y.}\ \bibnamefont
  {Liu}}, \bibinfo {author} {\bibfnamefont {X.}~\bibnamefont {Lin}}, \bibinfo
  {author} {\bibfnamefont {Z.~Y.}\ \bibnamefont {Chen}}, \bibinfo {author}
  {\bibfnamefont {C.}~\bibnamefont {Xue}}, \bibinfo {author} {\bibfnamefont
  {T.~P.}\ \bibnamefont {Sun}}, \bibinfo {author} {\bibfnamefont {Q.-S.}\
  \bibnamefont {Li}}, \bibinfo {author} {\bibfnamefont {X.~N.}\ \bibnamefont
  {Zhuang}}, \bibinfo {author} {\bibfnamefont {Y.~J.}\ \bibnamefont {Wang}},
  \bibinfo {author} {\bibfnamefont {Y.~C.}\ \bibnamefont {Wu}}, \bibinfo
  {author} {\bibfnamefont {M.}~\bibnamefont {Gong}},\ and\ \bibinfo {author}
  {\bibfnamefont {G.~P.}\ \bibnamefont {Guo}},\ }\bibfield  {title} {\bibinfo
  {title} {Simulation of open quantum systems on universal quantum computers},\
  }\href {https://doi.org/10.22331/q-2025-06-05-1765} {\bibfield  {journal}
  {\bibinfo  {journal} {{Quantum}}\ }\textbf {\bibinfo {volume} {9}},\ \bibinfo
  {pages} {1765} (\bibinfo {year} {2025})}\BibitemShut {NoStop}%
\bibitem [{\citenamefont {Lidar}\ and\ \citenamefont {Brun}(2013)}]{LidBru13}%
  \BibitemOpen
  \bibfield  {author} {\bibinfo {author} {\bibfnamefont {D.~A.}\ \bibnamefont
  {Lidar}}\ and\ \bibinfo {author} {\bibfnamefont {T.~A.}\ \bibnamefont
  {Brun}},\ }\href
  {https://www.cambridge.org/core/product/B51E8333050A0F9A67363254DC1EA15A}
  {\emph {\bibinfo {title} {Quantum Error Correction}}}\ (\bibinfo  {publisher}
  {Cambridge},\ \bibinfo {year} {2013})\BibitemShut {NoStop}%
\bibitem [{\citenamefont {Terhal}(2015)}]{Terhal15}%
  \BibitemOpen
  \bibfield  {author} {\bibinfo {author} {\bibfnamefont {B.~M.}\ \bibnamefont
  {Terhal}},\ }\bibfield  {title} {\bibinfo {title} {Quantum error correction
  for quantum memories},\ }\href {https://doi.org/10.1103/RevModPhys.87.307}
  {\bibfield  {journal} {\bibinfo  {journal} {Rev. Mod. Phys.}\ }\textbf
  {\bibinfo {volume} {87}},\ \bibinfo {pages} {307} (\bibinfo {year}
  {2015})}\BibitemShut {NoStop}%
\bibitem [{\citenamefont {Len}\ and\ \citenamefont {Ng}(2018)}]{LooKho18}%
  \BibitemOpen
  \bibfield  {author} {\bibinfo {author} {\bibfnamefont {Y.~L.}\ \bibnamefont
  {Len}}\ and\ \bibinfo {author} {\bibfnamefont {H.~K.}\ \bibnamefont {Ng}},\
  }\bibfield  {title} {\bibinfo {title} {Open-system quantum error
  correction},\ }\href {https://doi.org/10.1103/PhysRevA.98.022307} {\bibfield
  {journal} {\bibinfo  {journal} {Phys. Rev. A}\ }\textbf {\bibinfo {volume}
  {98}},\ \bibinfo {pages} {022307} (\bibinfo {year} {2018})}\BibitemShut
  {NoStop}%
\bibitem [{\citenamefont {Lieu}\ \emph {et~al.}(2020)\citenamefont {Lieu},
  \citenamefont {Belyansky}, \citenamefont {Young}, \citenamefont {Lundgren},
  \citenamefont {Albert},\ and\ \citenamefont {Gorshkov}}]{LieGor20}%
  \BibitemOpen
  \bibfield  {author} {\bibinfo {author} {\bibfnamefont {S.}~\bibnamefont
  {Lieu}}, \bibinfo {author} {\bibfnamefont {R.}~\bibnamefont {Belyansky}},
  \bibinfo {author} {\bibfnamefont {J.~T.}\ \bibnamefont {Young}}, \bibinfo
  {author} {\bibfnamefont {R.}~\bibnamefont {Lundgren}}, \bibinfo {author}
  {\bibfnamefont {V.~V.}\ \bibnamefont {Albert}},\ and\ \bibinfo {author}
  {\bibfnamefont {A.~V.}\ \bibnamefont {Gorshkov}},\ }\bibfield  {title}
  {\bibinfo {title} {Symmetry breaking and error correction in open quantum
  systems},\ }\href {https://doi.org/10.1103/PhysRevLett.125.240405} {\bibfield
   {journal} {\bibinfo  {journal} {Phys. Rev. Lett.}\ }\textbf {\bibinfo
  {volume} {125}},\ \bibinfo {pages} {240405} (\bibinfo {year}
  {2020})}\BibitemShut {NoStop}%
\bibitem [{\citenamefont {Spohn}(1978)}]{Spohn78}%
  \BibitemOpen
  \bibfield  {author} {\bibinfo {author} {\bibfnamefont {H.}~\bibnamefont
  {Spohn}},\ }\bibfield  {title} {\bibinfo {title} {Entropy production for
  quantum dynamical semigroups},\ }\href {https://doi.org/10.1063/1.523789}
  {\bibfield  {journal} {\bibinfo  {journal} {Journal of Mathematical Physics}\
  }\textbf {\bibinfo {volume} {19}},\ \bibinfo {pages} {1227} (\bibinfo {year}
  {1978})}\BibitemShut {NoStop}%
\bibitem [{\citenamefont {Wolf}\ and\ \citenamefont {Cirac}(2008)}]{WolCir08}%
  \BibitemOpen
  \bibfield  {author} {\bibinfo {author} {\bibfnamefont {M.~M.}\ \bibnamefont
  {Wolf}}\ and\ \bibinfo {author} {\bibfnamefont {J.~I.}\ \bibnamefont
  {Cirac}},\ }\bibfield  {title} {\bibinfo {title} {Dividing quantum
  channels},\ }\href {https://doi.org/10.1007/s00220-008-0411-y} {\bibfield
  {journal} {\bibinfo  {journal} {Communications in Mathematical Physics}\
  }\textbf {\bibinfo {volume} {279}},\ \bibinfo {pages} {147} (\bibinfo {year}
  {2008})}\BibitemShut {NoStop}%
\bibitem [{\citenamefont {Zanardi}\ and\ \citenamefont
  {Rasetti}(1997)}]{ZanRas97}%
  \BibitemOpen
  \bibfield  {author} {\bibinfo {author} {\bibfnamefont {P.}~\bibnamefont
  {Zanardi}}\ and\ \bibinfo {author} {\bibfnamefont {M.}~\bibnamefont
  {Rasetti}},\ }\bibfield  {title} {\bibinfo {title} {Noiseless quantum
  codes},\ }\href {https://doi.org/10.1103/PhysRevLett.79.3306} {\bibfield
  {journal} {\bibinfo  {journal} {Phys. Rev. Lett.}\ }\textbf {\bibinfo
  {volume} {79}},\ \bibinfo {pages} {3306} (\bibinfo {year}
  {1997})}\BibitemShut {NoStop}%
\bibitem [{\citenamefont {Lidar}\ \emph {et~al.}(1998)\citenamefont {Lidar},
  \citenamefont {Chuang},\ and\ \citenamefont {Whaley}}]{LidWha98}%
  \BibitemOpen
  \bibfield  {author} {\bibinfo {author} {\bibfnamefont {D.~A.}\ \bibnamefont
  {Lidar}}, \bibinfo {author} {\bibfnamefont {I.~L.}\ \bibnamefont {Chuang}},\
  and\ \bibinfo {author} {\bibfnamefont {K.~B.}\ \bibnamefont {Whaley}},\
  }\bibfield  {title} {\bibinfo {title} {Decoherence-free subspaces for quantum
  computation},\ }\href {https://doi.org/10.1103/PhysRevLett.81.2594}
  {\bibfield  {journal} {\bibinfo  {journal} {Phys. Rev. Lett.}\ }\textbf
  {\bibinfo {volume} {81}},\ \bibinfo {pages} {2594} (\bibinfo {year}
  {1998})}\BibitemShut {NoStop}%
\bibitem [{\citenamefont {Sasso}\ and\ \citenamefont
  {Umanità}(2023)}]{SasUma23}%
  \BibitemOpen
  \bibfield  {author} {\bibinfo {author} {\bibfnamefont {E.}~\bibnamefont
  {Sasso}}\ and\ \bibinfo {author} {\bibfnamefont {V.}~\bibnamefont
  {Umanità}},\ }\bibfield  {title} {\bibinfo {title} {The general structure of
  the decoherence-free subalgebra for uniformly continuous quantum {M}arkov
  semigroups},\ }\href {https://doi.org/10.1063/5.0092998} {\bibfield
  {journal} {\bibinfo  {journal} {Journal of Mathematical Physics}\ }\textbf
  {\bibinfo {volume} {64}},\ \bibinfo {pages} {042703} (\bibinfo {year}
  {2023})}\BibitemShut {NoStop}%
\bibitem [{\citenamefont {Li}\ \emph {et~al.}(2023)\citenamefont {Li},
  \citenamefont {Sala},\ and\ \citenamefont {Pollmann}}]{LiPol23}%
  \BibitemOpen
  \bibfield  {author} {\bibinfo {author} {\bibfnamefont {Y.}~\bibnamefont
  {Li}}, \bibinfo {author} {\bibfnamefont {P.}~\bibnamefont {Sala}},\ and\
  \bibinfo {author} {\bibfnamefont {F.}~\bibnamefont {Pollmann}},\ }\bibfield
  {title} {\bibinfo {title} {Hilbert space fragmentation in open quantum
  systems},\ }\href {https://doi.org/10.1103/PhysRevResearch.5.043239}
  {\bibfield  {journal} {\bibinfo  {journal} {Phys. Rev. Res.}\ }\textbf
  {\bibinfo {volume} {5}},\ \bibinfo {pages} {043239} (\bibinfo {year}
  {2023})}\BibitemShut {NoStop}%
\bibitem [{\citenamefont {Yoshida}(2024)}]{Yoshida24}%
  \BibitemOpen
  \bibfield  {author} {\bibinfo {author} {\bibfnamefont {H.}~\bibnamefont
  {Yoshida}},\ }\bibfield  {title} {\bibinfo {title} {Uniqueness of steady
  states of {G}orini-{K}ossakowski-{S}udarshan-{L}indblad equations: A simple
  proof},\ }\href {https://doi.org/10.1103/PhysRevA.109.022218} {\bibfield
  {journal} {\bibinfo  {journal} {Phys. Rev. A}\ }\textbf {\bibinfo {volume}
  {109}},\ \bibinfo {pages} {022218} (\bibinfo {year} {2024})}\BibitemShut
  {NoStop}%
\bibitem [{\citenamefont {Li}\ \emph {et~al.}(2025)\citenamefont {Li},
  \citenamefont {Liu}, \citenamefont {Guo}, \citenamefont {Zhang},
  \citenamefont {Zhao}, \citenamefont {Wang}, \citenamefont {Xiang},
  \citenamefont {Song}, \citenamefont {Zhang}, \citenamefont {Xu},
  \citenamefont {Fan},\ and\ \citenamefont {Zheng}}]{LiZhe25}%
  \BibitemOpen
  \bibfield  {author} {\bibinfo {author} {\bibfnamefont {L.}~\bibnamefont
  {Li}}, \bibinfo {author} {\bibfnamefont {T.}~\bibnamefont {Liu}}, \bibinfo
  {author} {\bibfnamefont {X.~Y.}\ \bibnamefont {Guo}}, \bibinfo {author}
  {\bibfnamefont {H.}~\bibnamefont {Zhang}}, \bibinfo {author} {\bibfnamefont
  {S.}~\bibnamefont {Zhao}}, \bibinfo {author} {\bibfnamefont {Z.~A.}\
  \bibnamefont {Wang}}, \bibinfo {author} {\bibfnamefont {Z.}~\bibnamefont
  {Xiang}}, \bibinfo {author} {\bibfnamefont {X.}~\bibnamefont {Song}},
  \bibinfo {author} {\bibfnamefont {Y.~X.}\ \bibnamefont {Zhang}}, \bibinfo
  {author} {\bibfnamefont {K.}~\bibnamefont {Xu}}, \bibinfo {author}
  {\bibfnamefont {H.}~\bibnamefont {Fan}},\ and\ \bibinfo {author}
  {\bibfnamefont {D.}~\bibnamefont {Zheng}},\ }\bibfield  {title} {\bibinfo
  {title} {Observation of multiple steady states with engineered dissipation},\
  }\href {https://doi.org/10.1038/s41534-025-00958-6} {\bibfield  {journal}
  {\bibinfo  {journal} {npj Quantum Information}\ }\textbf {\bibinfo {volume}
  {11}},\ \bibinfo {pages} {2} (\bibinfo {year} {2025})}\BibitemShut {NoStop}%
\bibitem [{\citenamefont {Baumgartner}\ \emph {et~al.}(2008)\citenamefont
  {Baumgartner}, \citenamefont {Narnhofer},\ and\ \citenamefont
  {Thirring}}]{BauThi08}%
  \BibitemOpen
  \bibfield  {author} {\bibinfo {author} {\bibfnamefont {B.}~\bibnamefont
  {Baumgartner}}, \bibinfo {author} {\bibfnamefont {H.}~\bibnamefont
  {Narnhofer}},\ and\ \bibinfo {author} {\bibfnamefont {W.}~\bibnamefont
  {Thirring}},\ }\bibfield  {title} {\bibinfo {title} {Analysis of quantum
  semigroups with {GKS}–{L}indblad generators: I. simple generators},\ }\href
  {https://doi.org/10.1088/1751-8113/41/6/065201} {\bibfield  {journal}
  {\bibinfo  {journal} {Journal of Physics A: Mathematical and Theoretical}\
  }\textbf {\bibinfo {volume} {41}},\ \bibinfo {pages} {065201} (\bibinfo
  {year} {2008})}\BibitemShut {NoStop}%
\bibitem [{\citenamefont {Baumgartner}\ and\ \citenamefont
  {Narnhofer}(2012)}]{BauNar12}%
  \BibitemOpen
  \bibfield  {author} {\bibinfo {author} {\bibfnamefont {B.}~\bibnamefont
  {Baumgartner}}\ and\ \bibinfo {author} {\bibfnamefont {H.}~\bibnamefont
  {Narnhofer}},\ }\bibfield  {title} {\bibinfo {title} {The structures of state
  space concerning quantum dynamical semigroups},\ }\href
  {https://doi.org/10.1142/S0129055X12500018} {\bibfield  {journal} {\bibinfo
  {journal} {Reviews in Mathematical Physics}\ }\textbf {\bibinfo {volume}
  {24}},\ \bibinfo {pages} {1250001} (\bibinfo {year} {2012})}\BibitemShut
  {NoStop}%
\bibitem [{\citenamefont {Albert}\ \emph {et~al.}(2016)\citenamefont {Albert},
  \citenamefont {Bradlyn}, \citenamefont {Fraas},\ and\ \citenamefont
  {Jiang}}]{AlbJia16}%
  \BibitemOpen
  \bibfield  {author} {\bibinfo {author} {\bibfnamefont {V.~V.}\ \bibnamefont
  {Albert}}, \bibinfo {author} {\bibfnamefont {B.}~\bibnamefont {Bradlyn}},
  \bibinfo {author} {\bibfnamefont {M.}~\bibnamefont {Fraas}},\ and\ \bibinfo
  {author} {\bibfnamefont {L.}~\bibnamefont {Jiang}},\ }\bibfield  {title}
  {\bibinfo {title} {Geometry and response of {L}indbladians},\ }\href
  {https://doi.org/10.1103/PhysRevX.6.041031} {\bibfield  {journal} {\bibinfo
  {journal} {Phys. Rev. X}\ }\textbf {\bibinfo {volume} {6}},\ \bibinfo {pages}
  {041031} (\bibinfo {year} {2016})}\BibitemShut {NoStop}%
\bibitem [{\citenamefont {Childs}\ \emph {et~al.}(2021)\citenamefont {Childs},
  \citenamefont {Su}, \citenamefont {Tran}, \citenamefont {Wiebe},\ and\
  \citenamefont {Zhu}}]{ChiZhu21}%
  \BibitemOpen
  \bibfield  {author} {\bibinfo {author} {\bibfnamefont {A.~M.}\ \bibnamefont
  {Childs}}, \bibinfo {author} {\bibfnamefont {Y.}~\bibnamefont {Su}}, \bibinfo
  {author} {\bibfnamefont {M.~C.}\ \bibnamefont {Tran}}, \bibinfo {author}
  {\bibfnamefont {N.}~\bibnamefont {Wiebe}},\ and\ \bibinfo {author}
  {\bibfnamefont {S.}~\bibnamefont {Zhu}},\ }\bibfield  {title} {\bibinfo
  {title} {Theory of trotter error with commutator scaling},\ }\href
  {https://doi.org/10.1103/PhysRevX.11.011020} {\bibfield  {journal} {\bibinfo
  {journal} {Phys. Rev. X}\ }\textbf {\bibinfo {volume} {11}},\ \bibinfo
  {pages} {011020} (\bibinfo {year} {2021})}\BibitemShut {NoStop}%
\bibitem [{\citenamefont {Kliesch}\ \emph {et~al.}(2011)\citenamefont
  {Kliesch}, \citenamefont {Barthel}, \citenamefont {Gogolin}, \citenamefont
  {Kastoryano},\ and\ \citenamefont {Eisert}}]{KliEis11}%
  \BibitemOpen
  \bibfield  {author} {\bibinfo {author} {\bibfnamefont {M.}~\bibnamefont
  {Kliesch}}, \bibinfo {author} {\bibfnamefont {T.}~\bibnamefont {Barthel}},
  \bibinfo {author} {\bibfnamefont {C.}~\bibnamefont {Gogolin}}, \bibinfo
  {author} {\bibfnamefont {M.}~\bibnamefont {Kastoryano}},\ and\ \bibinfo
  {author} {\bibfnamefont {J.}~\bibnamefont {Eisert}},\ }\bibfield  {title}
  {\bibinfo {title} {Dissipative quantum {C}hurch-{T}uring theorem},\ }\href
  {https://doi.org/10.1103/PhysRevLett.107.120501} {\bibfield  {journal}
  {\bibinfo  {journal} {Phys. Rev. Lett.}\ }\textbf {\bibinfo {volume} {107}},\
  \bibinfo {pages} {120501} (\bibinfo {year} {2011})}\BibitemShut {NoStop}%
\bibitem [{\citenamefont {Berry}\ \emph {et~al.}(2015)\citenamefont {Berry},
  \citenamefont {Childs}, \citenamefont {Cleve}, \citenamefont {Kothari},\ and\
  \citenamefont {Somma}}]{BerSom15}%
  \BibitemOpen
  \bibfield  {author} {\bibinfo {author} {\bibfnamefont {D.~W.}\ \bibnamefont
  {Berry}}, \bibinfo {author} {\bibfnamefont {A.~M.}\ \bibnamefont {Childs}},
  \bibinfo {author} {\bibfnamefont {R.}~\bibnamefont {Cleve}}, \bibinfo
  {author} {\bibfnamefont {R.}~\bibnamefont {Kothari}},\ and\ \bibinfo {author}
  {\bibfnamefont {R.~D.}\ \bibnamefont {Somma}},\ }\bibfield  {title} {\bibinfo
  {title} {Simulating {H}amiltonian dynamics with a truncated {T}aylor
  series},\ }\href {https://doi.org/10.1103/PhysRevLett.114.090502} {\bibfield
  {journal} {\bibinfo  {journal} {Phys. Rev. Lett.}\ }\textbf {\bibinfo
  {volume} {114}},\ \bibinfo {pages} {090502} (\bibinfo {year}
  {2015})}\BibitemShut {NoStop}%
\bibitem [{\citenamefont {Trefethen}\ and\ \citenamefont
  {Embree}(2005)}]{TreEmb05}%
  \BibitemOpen
  \bibfield  {author} {\bibinfo {author} {\bibfnamefont {L.~N.}\ \bibnamefont
  {Trefethen}}\ and\ \bibinfo {author} {\bibfnamefont {M.}~\bibnamefont
  {Embree}},\ }\href@noop {} {\emph {\bibinfo {title} {Spectra and
  Pseudospectra: The Behavior of Nonnormal Matrices and Operators}}}\ (\bibinfo
   {publisher} {Princeton University Press},\ \bibinfo {address} {Princeton,
  NJ},\ \bibinfo {year} {2005})\BibitemShut {NoStop}%
\bibitem [{\citenamefont {Chalker}\ and\ \citenamefont
  {Mehlig}(1998)}]{ChaMeh98}%
  \BibitemOpen
  \bibfield  {author} {\bibinfo {author} {\bibfnamefont {J.~T.}\ \bibnamefont
  {Chalker}}\ and\ \bibinfo {author} {\bibfnamefont {B.}~\bibnamefont
  {Mehlig}},\ }\bibfield  {title} {\bibinfo {title} {Eigenvector statistics in
  non-{H}ermitian random matrix ensembles},\ }\href
  {https://doi.org/10.1103/PhysRevLett.81.3367} {\bibfield  {journal} {\bibinfo
   {journal} {Phys. Rev. Lett.}\ }\textbf {\bibinfo {volume} {81}},\ \bibinfo
  {pages} {3367} (\bibinfo {year} {1998})}\BibitemShut {NoStop}%
\bibitem [{\citenamefont {Naves}\ \emph {et~al.}(2026)\citenamefont {Naves},
  \citenamefont {Kvorning},\ and\ \citenamefont {Larson}}]{NavLar26}%
  \BibitemOpen
  \bibfield  {author} {\bibinfo {author} {\bibfnamefont {C.~B.}\ \bibnamefont
  {Naves}}, \bibinfo {author} {\bibfnamefont {T.~K.}\ \bibnamefont
  {Kvorning}},\ and\ \bibinfo {author} {\bibfnamefont {J.}~\bibnamefont
  {Larson}},\ }\bibfield  {title} {\bibinfo {title} {When level repulsion
  fails: non-normality and chaos in open quantum systems},\ }\href
  {https://arxiv.org/abs/2604.00118} {\bibfield  {journal} {\bibinfo  {journal}
  {arXiv:2604.00118}\ } (\bibinfo {year} {2026})}\BibitemShut {NoStop}%
\bibitem [{\citenamefont {Ashida}\ \emph {et~al.}(2020)\citenamefont {Ashida},
  \citenamefont {Gong},\ and\ \citenamefont {Ueda}}]{AshUed20}%
  \BibitemOpen
  \bibfield  {author} {\bibinfo {author} {\bibfnamefont {Y.}~\bibnamefont
  {Ashida}}, \bibinfo {author} {\bibfnamefont {Z.}~\bibnamefont {Gong}},\ and\
  \bibinfo {author} {\bibfnamefont {M.}~\bibnamefont {Ueda}},\ }\bibfield
  {title} {\bibinfo {title} {Non-{H}ermitian physics},\ }\href
  {https://doi.org/10.1080/00018732.2021.1876991} {\bibfield  {journal}
  {\bibinfo  {journal} {Advances in Physics}\ }\textbf {\bibinfo {volume}
  {69}},\ \bibinfo {pages} {249} (\bibinfo {year} {2020})}\BibitemShut
  {NoStop}%
\bibitem [{\citenamefont {Bello}\ \emph {et~al.}(2025)\citenamefont {Bello},
  \citenamefont {Pavan}, \citenamefont {Cataudella},\ and\ \citenamefont
  {Farina}}]{BelDon25}%
  \BibitemOpen
  \bibfield  {author} {\bibinfo {author} {\bibfnamefont {G.~D.}\ \bibnamefont
  {Bello}}, \bibinfo {author} {\bibfnamefont {F.}~\bibnamefont {Pavan}},
  \bibinfo {author} {\bibfnamefont {V.}~\bibnamefont {Cataudella}},\ and\
  \bibinfo {author} {\bibfnamefont {D.}~\bibnamefont {Farina}},\ }\bibfield
  {title} {\bibinfo {title} {Local and global master equations through the lens
  of non-{H}ermitian physics},\ }\href {https://arxiv.org/abs/2509.10425}
  {\bibfield  {journal} {\bibinfo  {journal} {arXiv:2509.10425}\ } (\bibinfo
  {year} {2025})}\BibitemShut {NoStop}%
\bibitem [{\citenamefont {Lindblad}(1976)}]{Lindblad76}%
  \BibitemOpen
  \bibfield  {author} {\bibinfo {author} {\bibfnamefont {G.}~\bibnamefont
  {Lindblad}},\ }\bibfield  {title} {\bibinfo {title} {On the generators of
  quantum dynamical semigroups},\ }\href {https://doi.org/10.1007/BF01608499}
  {\bibfield  {journal} {\bibinfo  {journal} {Communications in Mathematical
  Physics}\ }\textbf {\bibinfo {volume} {48}},\ \bibinfo {pages} {119}
  (\bibinfo {year} {1976})}\BibitemShut {NoStop}%
\bibitem [{\citenamefont {Gorini}\ \emph {et~al.}(1976)\citenamefont {Gorini},
  \citenamefont {Kossakowski},\ and\ \citenamefont {Sudarshan}}]{GKS76}%
  \BibitemOpen
  \bibfield  {author} {\bibinfo {author} {\bibfnamefont {V.}~\bibnamefont
  {Gorini}}, \bibinfo {author} {\bibfnamefont {A.}~\bibnamefont
  {Kossakowski}},\ and\ \bibinfo {author} {\bibfnamefont {E.~C.~G.}\
  \bibnamefont {Sudarshan}},\ }\bibfield  {title} {\bibinfo {title} {Completely
  positive dynamical semigroups of n‐level systems},\ }\href
  {https://doi.org/10.1063/1.522979} {\bibfield  {journal} {\bibinfo  {journal}
  {Journal of Mathematical Physics}\ }\textbf {\bibinfo {volume} {17}},\
  \bibinfo {pages} {821} (\bibinfo {year} {1976})}\BibitemShut {NoStop}%
\bibitem [{\citenamefont {Gronwall}(1919)}]{Gronwall1919}%
  \BibitemOpen
  \bibfield  {author} {\bibinfo {author} {\bibfnamefont {T.~H.}\ \bibnamefont
  {Gronwall}},\ }\bibfield  {title} {\bibinfo {title} {Note on the derivatives
  with respect to a parameter of the solutions of a system of differential
  equations},\ }\href {http://www.jstor.org/stable/1967124} {\bibfield
  {journal} {\bibinfo  {journal} {Annals of Mathematics}\ }\textbf {\bibinfo
  {volume} {20}},\ \bibinfo {pages} {292} (\bibinfo {year} {1919})}\BibitemShut
  {NoStop}%
\bibitem [{\citenamefont {Low}\ and\ \citenamefont {Chuang}(2017)}]{LowChu17}%
  \BibitemOpen
  \bibfield  {author} {\bibinfo {author} {\bibfnamefont {G.~H.}\ \bibnamefont
  {Low}}\ and\ \bibinfo {author} {\bibfnamefont {I.~L.}\ \bibnamefont
  {Chuang}},\ }\bibfield  {title} {\bibinfo {title} {Optimal {H}amiltonian
  simulation by quantum signal processing},\ }\href
  {https://doi.org/10.1103/PhysRevLett.118.010501} {\bibfield  {journal}
  {\bibinfo  {journal} {Phys. Rev. Lett.}\ }\textbf {\bibinfo {volume} {118}},\
  \bibinfo {pages} {010501} (\bibinfo {year} {2017})}\BibitemShut {NoStop}%
\bibitem [{\citenamefont {Borras}\ and\ \citenamefont
  {Marvian}(2025)}]{BorMar25}%
  \BibitemOpen
  \bibfield  {author} {\bibinfo {author} {\bibfnamefont {E.}~\bibnamefont
  {Borras}}\ and\ \bibinfo {author} {\bibfnamefont {M.}~\bibnamefont
  {Marvian}},\ }\bibfield  {title} {\bibinfo {title} {Quantum algorithms based
  on quantum trajectories},\ }\href {https://arxiv.org/abs/2509.10425}
  {\bibfield  {journal} {\bibinfo  {journal} {arXiv:2509.10425}\ } (\bibinfo
  {year} {2025})}\BibitemShut {NoStop}%
\bibitem [{\citenamefont {Berry}\ \emph {et~al.}(2007)\citenamefont {Berry},
  \citenamefont {Ahokas}, \citenamefont {Cleve},\ and\ \citenamefont
  {Sanders}}]{BerSan07}%
  \BibitemOpen
  \bibfield  {author} {\bibinfo {author} {\bibfnamefont {D.~W.}\ \bibnamefont
  {Berry}}, \bibinfo {author} {\bibfnamefont {G.}~\bibnamefont {Ahokas}},
  \bibinfo {author} {\bibfnamefont {R.}~\bibnamefont {Cleve}},\ and\ \bibinfo
  {author} {\bibfnamefont {B.~C.}\ \bibnamefont {Sanders}},\ }\bibfield
  {title} {\bibinfo {title} {Efficient quantum algorithms for simulating sparse
  {H}amiltonians},\ }\href {https://doi.org/10.1007/s00220-006-0150-x}
  {\bibfield  {journal} {\bibinfo  {journal} {Communications in Mathematical
  Physics}\ }\textbf {\bibinfo {volume} {270}},\ \bibinfo {pages} {359}
  (\bibinfo {year} {2007})}\BibitemShut {NoStop}%
\bibitem [{\citenamefont {Childs}\ and\ \citenamefont {van
  Dam}(2010)}]{ChiDam10}%
  \BibitemOpen
  \bibfield  {author} {\bibinfo {author} {\bibfnamefont {A.~M.}\ \bibnamefont
  {Childs}}\ and\ \bibinfo {author} {\bibfnamefont {W.}~\bibnamefont {van
  Dam}},\ }\bibfield  {title} {\bibinfo {title} {Quantum algorithms for
  algebraic problems},\ }\href {https://doi.org/10.1103/RevModPhys.82.1}
  {\bibfield  {journal} {\bibinfo  {journal} {Rev. Mod. Phys.}\ }\textbf
  {\bibinfo {volume} {82}},\ \bibinfo {pages} {1} (\bibinfo {year}
  {2010})}\BibitemShut {NoStop}%
\bibitem [{\citenamefont {Childs}\ \emph {et~al.}(2018)\citenamefont {Childs},
  \citenamefont {Maslov}, \citenamefont {Nam}, \citenamefont {Ross},\ and\
  \citenamefont {Su}}]{ChiSu18}%
  \BibitemOpen
  \bibfield  {author} {\bibinfo {author} {\bibfnamefont {A.~M.}\ \bibnamefont
  {Childs}}, \bibinfo {author} {\bibfnamefont {D.}~\bibnamefont {Maslov}},
  \bibinfo {author} {\bibfnamefont {Y.}~\bibnamefont {Nam}}, \bibinfo {author}
  {\bibfnamefont {N.~J.}\ \bibnamefont {Ross}},\ and\ \bibinfo {author}
  {\bibfnamefont {Y.}~\bibnamefont {Su}},\ }\bibfield  {title} {\bibinfo
  {title} {Toward the first quantum simulation with quantum speedup},\ }\href
  {https://doi.org/10.1073/pnas.1801723115} {\bibfield  {journal} {\bibinfo
  {journal} {Proceedings of the National Academy of Sciences}\ }\textbf
  {\bibinfo {volume} {115}},\ \bibinfo {pages} {9456} (\bibinfo {year}
  {2018})}\BibitemShut {NoStop}%
\end{thebibliography}%
	
\end{document}